\documentclass[11pt]{article}
\usepackage[left=1in,top=1in,right=1in,bottom=1in,head=.1in,nofoot]{geometry}

\usepackage{setspace,url,bm,amsmath} 
\usepackage{float}
\usepackage{caption}
\usepackage{subcaption,siunitx,booktabs}

\usepackage{titlesec} 
\usepackage{rotating}
\titlelabel{\thetitle.\quad} 
\titleformat*{\section}{\bf\large}

\usepackage{graphicx} 
\usepackage{bbm}
\usepackage{latexsym}
\usepackage{color}

\usepackage{amsthm}
\usepackage{amsfonts}
\theoremstyle{definition}
\newtheorem{proposition}{Proposition}
\newtheorem{lemma}{Lemma}
\newtheorem{example}{Example}
\newtheorem{corollary}{Corollary}

\usepackage{natbib} 
\bibpunct{(}{)}{;}{a}{}{,} 
\usepackage{enumitem}

\usepackage{etoolbox} 
\apptocmd{\sloppy}{\hbadness 10000\relax}{}{} 

\def\ind{\begin{picture}(9,8)
         \put(0,0){\line(1,0){9}}
         \put(3,0){\line(0,1){8}}
         \put(6,0){\line(0,1){8}}
         \end{picture}
        }

\begin{document}

\doublespacing
\frenchspacing
\title{\bf 
Treatment effects on ordinal outcomes: Causal estimands and sharp bounds}
\author{Jiannan Lu, Peng Ding~and Tirthankar Dasgupta\thanks{Jiannan Lu is Data Scientist (E-mail: jiannl@microsoft.com), Analysis and Experimentation, Microsoft Corporation, Redmond, WA 98052, U.S.A. Peng Ding is Assistant Professor (E-mail: pengdingpku@berkeley.edu), Department of Statistics, University California, Berkeley, CA 94270, U.S.A. Tirthankar Dasgupta is Associate Professor (E-mail: tirthankar.dasgupta@rutgers.edu), Department of Statistics and Biostatistics, Rutgers University, Piscataway, NJ 08854-8019, U.S.A. 
}}

\date{}
\maketitle

\begin{abstract}
\frenchspacing
Assessing the causal effects of interventions on ordinal outcomes is an important objective of many educational and behavioral studies. Under the potential outcomes framework, we can define causal effects as comparisons between the potential outcomes under treatment and control. However, unfortunately, the average causal effect, often the parameter of interest, is difficult to interpret for ordinal outcomes. To address this challenge, we propose to use two causal parameters, which are defined as the probabilities that the treatment is beneficial and strictly beneficial for the experimental units. However, although well-defined for any outcomes and of particular interest for ordinal outcomes, the two aforementioned parameters depend on the association between the potential outcomes, and are therefore not identifiable from the observed data without additional assumptions. Echoing recent advances in the econometrics and biostatistics literature, we present the sharp bounds of the aforementioned causal parameters for ordinal outcomes, under fixed marginal distributions of the potential outcomes. Because the causal estimands and their corresponding sharp bounds are based on the potential outcomes themselves, the proposed framework can be flexibly incorporated into any chosen models of the potential outcomes, and are directly applicable to randomized experiments, unconfounded observational studies, and randomized experiments with noncompliance. We illustrate our methodology via numerical examples and three real-life applications related to educational and behavioral research.
\end{abstract}


\section{Introduction}
\frenchspacing

In educational, behavioral and public health research, a scenario frequently encountered is evaluating causal effects of interventions on ordinal (i.e., ordered categorical) outcomes. For example, \cite{Oenema:2001} conducted a randomized controlled trial to access whether web-based nutrition education changed personal awareness and intentions (e.g. negative, neutral or positive attitudes) towards healthier diets. \cite{Hoff:2009} analyzed a data-set from the 1994 General Social Survey \citep{Smith:2013}, aiming to study whether the fact that parents possessing college or higher degrees affected their offspring's education level (from ``less than high school'' to ``graduate degree''). \cite{Praet:2014} investigated the effect of computer-aided programs on young children's proficiency in arithmetic (e.g., 0--10 scaled scores in reading, writing and counting). To draw scientifically meaningful conclusions from such studies, it is imperative that we employ an interpretable and robust methodology for defining and inferring causal effects.

The potential outcomes framework \citep{Neyman:1923, Rubin:1974} permits defining causal effects as comparisons between the potential outcomes under treatment and control. The average causal effect, generally the parameter of interest ever since the seminal work of \citet{Neyman:1923}, may not be applicable to ordinal outcomes, because average outcomes themselves are not well-defined substantively (although they can be well-defined mathematically),{
\color{black} except when there are meaningful distances between outcomes (e.g., standard test scores).
}For example, it is difficult to interpret the ``average'' of ``high school'' and ``Ph.D.,'' or compare it to the ``average'' of ``bachelor'' and ``master.'' Nevertheless, ordinal outcomes appear rather frequently in applied research, and the generalized linear model literature \citep[cf.][]{Agresti:2010} has discussed them extensively. However, although the model parameters of the generalized linear models are useful summaries of the data, they are often not direct measures of the causal effects of interest \citep{Freedman:2008}. More importantly, statistical inference often requires correctly-specified models, and when the generalized linear model assumptions are violated, the interpretations of the parameters become obscure. Mainly focused on the classic average causal effect (and its variants), the existing causal inference literature does not thoroughly investigate ordinal outcomes. Exceptions include \cite{Rosenbaum:2001}, who discussed causal inference for ordinal outcomes under the monotonicity assumption that the treatment is beneficial for all units. \cite{Cheng:2009}, \cite{Agresti:2010} and \citet{Agresti:2017} discussed various causal parameters under the assumption of independent potential outcomes. \citet{Volfovsky:2015} exploited a Bayesian strategy, requiring a full parametric model on the joint values of the potential outcomes. \cite{Diaz:2016} proposed to use a causal parameter that did not rely on the assumption of the proportional odds model for ordinal outcomes.

Realizing the conceptual and theoretical gaps in this important topic, in this paper we propose to use two causal parameters for ordinal outcomes, measuring the probabilities that the treatment is beneficial and strictly beneficial for the experimental units. The two parameters play important roles in decision and policy making for randomized evaluations with ordinal outcomes. However, because the two causal parameters depend on the association between the treatment and control potential outcomes, they are generally not identifiable from the observed data. Instead of imposing assumptions about the underlying distributions of, or the association between, the potential outcomes, we adopt the partial identification philosophy \citep[c.f.][]{Manski:2003, Richardson:2014} and sharply bound the parameters by using the marginal distributions of the potential outcomes. We acknowledge concurrent work by \cite{Huang:2017}, who numerically calculated the sharp bounds of the parameters and provided their consistent estimators,{
\color{black} allowing for potentially complex support restrictions on the marginal distributions of the potential outcomes.
}Compared to \cite{Huang:2017}, one main distinction of our work is that we focus on the identification perspective. To be specific, echoing several relevant discussions in the discrete mathematics \citep{Williamson:1990} and econometrics \citep[e.g.,][]{Manski:1997, Manski:2009, Fan:2009, Kim:2014} literature, we present closed form expressions for the sharp bounds of the causal parameters. 

We believe that the mathematical practice of deriving the closed-form expressions for the sharp bounds has a two-fold benefit. From a theoretical perspective, the closed form expressions enable us to study when we can identify the causal parameters, i.e., the lower and upper bounds collapse. At least in the context of ordinal outcomes, we believe this is a unique contribution to the existing literature. From a more practical perspective, because these bounds are defined by the potential outcomes themselves, they can be incorporated flexibly into any chosen models of the potential outcomes in practice. Furthermore, they are directly applicable to randomized experiments, unconfounded observational studies, and randomized experiments with noncompliance. In randomized experiments, we can identify the bounds immediately, and additionally, sharpen the bounds by exploiting covariate information under certain modeling assumptions. In observational studies, if the treatment assignment is unconfounded given the observed covariates, we can identify the bounds, for example by the propensity score weighting \citep{Rosenbaum:1983, Hirano:2003}. Furthermore, we extend the theory to accommodate noncompliance, which often arises in practical randomized evaluations.

The paper proceeds as follows. Section \ref{sec:review} introduces the potential outcomes framework for causal inference for ordinal outcomes, and proposes two causal parameters that are natural measures of causal effects and are of practical importance. Section \ref{sec:theory} presents the sharp bounds of the proposed causal parameters. Section \ref{sec:noncompliance} generalizes the bounds to noncompliance. Section \ref{sec:inf} discusses statistical inference of the bounds. Sections \ref{sec:simu} and \ref{sec:example} present numerical and real examples to illustrate the theoretical results. We conclude in Section \ref{sec:discuss}, and give all the proofs, technical and computational details in the Supplementary Material.

\section{Causal Inference for Ordinal Outcomes}\label{sec:review}

\subsection{Potential outcomes}
We consider a study with $N$ units, a binary treatment, and an ordinal outcome with $J$ categories labeled as $0,\ldots,J-1,$ where 0 and $J-1$ represent the worst and best categories. Under the Stable Unit Treatment Value Assumption \citep{Rubin:1980} that there is only one version of the treatment and no interference among the units, we define the pair $\left\{Y_i(1), Y_i(0)\right\}$ as the potential outcomes of the $i$th unit under treatment and control, respectively. Let
$$
p_{kl} = \mathrm{pr} \left\{ Y_i(1) = k, Y_i(0) = l \right\} \quad (k,l=0, \ldots, J-1)
$$
denote the proportion or probability of units whose potential outcome is $k$ under treatment and $l$ under control. The probability notation ``$\mathrm{pr}(\cdot)$'' is either for a finite population of $N$ units or for a super population, depending on the question of interest. The probability matrix $\bm P = (p_{kl})_{0\leq k,l\leq J-1}$ summarizes the (unconditional) joint distribution of the potential outcomes. We denote the row and column sums of $\bm P$ by
\begin{equation*}
p_{k+} = \sum_{l^\prime=0}^{J-1}p_{kl^\prime},
\quad
p_{+l} = \sum_{k^\prime =0}^{J-1}p_{k^\prime l}
\quad
(k, l = 0, 1, \ldots, J-1).
\end{equation*}
The vectors $\bm p_1 = \left( p_{0+}, \ldots, p_{J-1, +} \right)^\textrm{T}$ and $\bm p_0 = \left( p_{+0}, \ldots, p_{+, J-1} \right)^\textrm{T}$ characterize the marginal distributions of the potential outcomes under treatment and control, respectively. By definition, the following constraints must hold: 
\begin{equation*}
\sum_{k=0}^{J-1} p_{k+} = 1,
\quad
\sum_{l=0}^{J-1}p_{+l} = 1,
\quad
\sum_{k=0}^{J-1} \sum_{l=0}^{J-1}p_{kl} = 1.
\end{equation*}

\subsection{Causal parameters for ordinal outcomes}

We discuss the existing causal parameters for ordinal outcomes, and the motivation behind proposing new ones. Any causal parameter is a function of the probability matrix $\bm P.$ Unfortunately, the average causal effect is difficult to interpret for ordinal outcomes. Instead, we can use the distributional causal effects \citep[cf.][]{Ju:2010}
\begin{equation}\label{eq:DCE}
\Delta_j = \mathrm{pr}\left\{ Y_i(1) \ge j \right\} - \mathrm{pr}\left\{ Y_i(0) \ge j \right\}
= \sum_{k\ge j} p_{k+}  -  \sum_{l\ge j}p_{+l}
\quad
(j = 0, \ldots, J-1)
\end{equation}
to measure the difference between the marginal distributions of potential outcomes at different levels of $j.${
\color{black} Although distributional causal effects are standard and important measures for ordinal outcomes in practice, it is sometimes difficult to decide whether the treatment or the control is preferable unless they have the same sign for all $j.$ In the presence of heterogeneous distributional treatment effects for different levels of $j,$  we may use
}$\sum_{j=1}^{J-1}\omega_j \Delta_j$ to measure the treatment effect, but such a measure depends crucially on the weights $\omega_j$'s. We illustrate this point by using the following numerical example.

\begin{example}\label{example:0}
Let
$
\bm p_1 = \left( 1/5, 3/5, 1/5 \right)^\textrm{T}
$
and
$
\bm p_0 = \left( 2/5, 1/5, 2/5 \right)^\textrm{T},
$
with $\Delta_0 = 0,$ $\Delta_1 = 1/5$ and $\Delta_2 = -1/5.$ The treatment is beneficial at level 1, but not at level 2. In this case,  distributional causal effects do not provide straightforward guidance for decision making.
\end{example}

When $\Delta_j \geq 0$ for all $j$, $Y(1)$ stochastically dominates $Y(0)$. When this pattern appears in real data applications, practitioners often fit a proportional odds model \citep{Agresti:2010} and summarize the overall effectiveness of the treatment by a single odds ratio parameter. Although such summary parameter may be useful in certain cases, its causal interpretation is unclear. Moreover, when the data does not present the stochastic dominance pattern as in Example \ref{example:0}, summarizing the treatment effect by the single odds ratio parameter of a wrong model often gives misleading conclusions.

\cite{Volfovsky:2015} studied the conditional medians
\begin{equation}\label{eq:conditional-median}
m_j
=\mathrm{med}\left\{Y_i( 1)\mid Y_i(0)=j\right\}
\quad
(j=0, \ldots, J-1),
\end{equation}
which is a set containing all values of $k$ such that
$
\sum_{k' = 0}^{k} p_{k' j} \ge p_{+j}/2
$
and
$
\sum_{k' = k}^{J-1} p_{k'j} \ge p_{+j}/2.
$
By definition, the conditional medians may not be unique, and they are only well-defined for $j$ with $p_{+j} >0.$ Moreover, they are not direct measures of the treatment effect itself.

We propose to use two causal parameters that measure the probabilities that the treatment is beneficial and strictly beneficial for the experimental units:
\begin{equation}\label{eq:obj}
\tau = \mathrm{pr} \left\{ Y_i (1) \ge Y_i (0) \right\} = \mathop{\sum\sum}_{ k \ge  l}p_{kl},
\quad
\eta = \mathrm{pr} \left\{ Y_i (1) > Y_i (0) \right\} = \mathop{\sum\sum}_{ k >  l}p_{kl}.
\end{equation}
The causal parameters $\tau$ and $\eta$ are measures of causal effects that are well-defined for any types of outcomes, and of particular interest to ordinal outcomes.{
\color{black} To be more specific, they can complement the distributional causal effects and provide more information about what would happen under treatment versus control for an ordinal outcome.
}Similar causal measures appeared in biomedical \citep{Gadbury:2000, Newcombe:2006a, Newcombe:2006b, Zhou:2008, Huang:2017, Demidenko:2016} and social sciences \citep{Heckman:1997, Djebbari:2008, Fan:2010, Fan:2014}. In practice, we suggest using the pair $(\tau, \eta)$ as measures of causal effects on ordinal outcomes. For example, if the sharp null holds, i.e., $Y_i(1) = Y_i(0)$ for all units $i$, then $\tau = 1$ and $\eta = 0$. In this case, using only $\tau$ may be misleading. Nevertheless, we argue that the parameter $\tau$ is as important as $\eta$. Because $ 1 - \tau =  \mathrm{pr} \left\{ Y_i (0) > Y_i (1) \right\}$, the value of $\tau$ determines the probability that the control is strictly beneficial for the experimental units. Due to the symmetry of treatment and control labels, $\tau$ and $\eta$ are equally useful for real data analysis.

We use the following numerical example to show the values of $m_j$, $\tau$ and $\eta.$

\begin{example}\label{example:toy}
Consider the following probability matrix:
\begin{equation*}
\bm P =
\left(
\begin{array}{ccc}
0 & 1/6 & 1/6 \\
0 & 1/6 & 0 \\
0 & 1/3 & 1/6 \\
\end{array} \right).
\end{equation*}
In this case, $m_0$ is not well-defined because $p_{k0}=0$ for all $k,$ $m_1$ is $1,$ and $m_2 = \{ 0, 1 , 2\}$ by the definition of the conditional median in \eqref{eq:conditional-median}. However, we have $\tau=2/3$ and $\eta=1/3,$ i.e., two thirds of the population benefit from the treatment and one third strictly benefit.
\end{example}

The causal parameters $\tau$ and $\eta$ in \eqref{eq:obj} are well-defined for both finite populations and super populations. They are functions of the potential outcomes, which distinguishes them from the parameters in super population models. When the models are mis-specified, the interpretations of the corresponding model parameters are often obscure. We have already discussed this issue for the proportional odds model. Our causal parameters $\tau$ and $\eta$ are closely related to the relative treatment effect 
$
\alpha = \mathrm{pr}\left\{ Y_i(1) > Y_i(0) \right\} - \mathrm{pr}\left\{ Y_i(1) < Y_i(0) \right\}
$
previously studied under the assumption of independent potential outcomes \citep{Agresti:2010}.
This relative treatment effect $\alpha$ and the causal parameters we proposed have a simple algebraic relationship, i.e., $\alpha=\tau+\eta-1.$ Therefore, our newly proposed causal parameters $\tau$ and $\eta$ determine $\alpha.$ Furthermore, these causal parameters are also related to the notation of ``probability of causation'' \citep{Pearl:2009}, because their direct interpretations are the probabilities or proportions that the treatment affects the outcome on the individual level. It is for these reasons that we advocate using $\tau$ and $\eta$ as causal effect measures for ordinal outcomes.

\section{Sharp Bounds on the Proposed Causal Estimands for Ordinal Outcomes}\label{sec:theory}

\subsection{Closed-form expressions of sharp bounds}

The definitions of $\tau$ and $\eta$ involve the association between the treatment and control potential outcomes. Because we can never jointly measure the potential outcomes, the observed data do not provide full information about their association, rendering the causal parameters $\tau$ and $\eta$ not identifiable. To partially circumvent this difficulty, we focus on the sharp bounds of $\tau$ and $\eta,$ which are the minimal and maximal values of $\tau$ and $\eta$ under the condition that the probability matrix
$
\bm P = (p_{kl})_{0\leq k,l\leq J-1}
$
is well-defined, as well as the constraints of the marginal distributions. In other words, the following needs to hold:
\begin{equation}\label{eq:linear-constraint}
\sum_{l^\prime = 0 }^{J-1} p_{k l^\prime} = p_{k +},
\quad
\sum_{k^\prime = 0}^{J-1}p_{k^\prime l} = p_{+ l},
\quad
p_{kl} \ge 0
\quad
(k, l = 0,\ldots,J-1).
\end{equation}
The sharp bounds depend only on the marginal distributions of the potential outcomes. Deriving the sharp bounds is equivalent to solving linear programming problems, because the objective functions in \eqref{eq:obj} and the constraints in \eqref{eq:linear-constraint} are all linear. Previous literature \citep{Huang:2017} used a numerical method to solve the linear programming problem for $\eta$. 
Fortunately, as pointed out by several researchers \citep{Williamson:1990, Fan:2009}, we can derive closed-form solutions of the above linear programming problems, for both $\tau$ and $\eta$. We first present the sharp bounds of $\tau,$ which is the foundation for the remaining of the paper.

\begin{proposition}\label{thm:tau}
The sharp lower and upper bound of $\tau$ are
\begin{equation}
\label{eq:tau}
\tau_L =\max\limits_{0 \le j \le J-1} \left( p_{+ j} + \Delta_j \right),
\quad
\tau_U = 1 + \min\limits_{0 \le j \le J-1} \Delta_j.
\end{equation}
\end{proposition}
{\color{black}
The bounds in \eqref{eq:tau} resemble \cite{Fan:2010}'s parallel results for continuous outcomes, where the maximum and minimum operators are replaced by supremum and infimum, respectively.
}As a straightforward validity check, note that the inequalities $0 \le \tau_L \le \tau_U \le 1$ always hold, regardless of the marginal distributions of the potential outcomes. In particular, by definition in \eqref{eq:tau}, $\tau_L \ge p_{+0} + \Delta_0 = p_{+0},$ and $\tau_U \le 1 + \Delta_0 = 1.$ Moreover, the bounds in Proposition \ref{thm:tau} are closely related to the distributional causal effects in \eqref{eq:DCE}, and therefore we can interpret them as the conservative and optimistic estimates of the probability that the treatment is beneficial to the outcome. Furthermore, the following corollary demonstrates that the sharp upper bound $\tau_U$ is related to the \emph{stochastic dominance} assumption, i.e.,
$
\Delta_j \ge 0
$
for all $j.$

\begin{corollary}\label{coro:tauu-and-SD}
The causal parameter $\tau_U = 1,$ if and only if the marginal probabilities $\bm p_1$ and $\bm p_0$ satisfy the stochastic dominance assumption.
\end{corollary}

The above corollary implies that for any marginal probabilities satisfying the stochastic dominance assumption, there exists a lower triangular probability matrix $\bm{P}$ that corresponds to a population satisfying the monotonicity assumption, i.e.,
$
Y_i(1) \ge Y_i(0)
$
for all $i.$
\citet{Strassen:1965} and \citet{Rosenbaum:2001} demonstrated this result, and Proposition \ref{thm:tau} extends the previous result without imposing the stochastic dominance assumption. Moreover, Proposition \ref{thm:tau} also justifies the use of $\min_{0 \le j \le J-1} \Delta_j$ as a measure of the deviation from the stochastic dominance assumption \citep{Scharfstein:2004}.

To bounds $\eta,$ realizing that
$
\eta = 1- \mathrm{pr}\left\{Y_i(0) \ge Y_i(1)\right\},
$
we can directly derive the sharp bounds for
$
\mathrm{pr}\left\{Y_i(0)\ge Y_i(1)\right\}
$
by switching the treatment and control labels and applying \eqref{eq:tau}.

\begin{proposition}\label{thm:eta}
The sharp lower and upper bounds of $\eta$ are
\begin{equation}\label{eq:eta}
\eta_L = \max\limits_{0 \le j \le J-1} \Delta_j,
\quad
\eta_U = 1 + \min_{0 \le j \le J-1} \left(\Delta_j - p_{j +}\right).
\end{equation}
\end{proposition}

Similarly as the sharp bounds for $\tau$ in \eqref{eq:tau}, the inequalities $0 \le \eta_L \le \eta_U \le 1$ always hold. The bounds in \eqref{eq:tau}--\eqref{eq:eta} resembles parallel results in the econometrics literature \citep{Manski:1997, Manski:2000, Fan:2009, Fan:2010}, which largely focused on continuous outcomes. In fact, deriving the sharp bounds of $\tau$ and $\eta$ is related to a classical probability problem posed by A. N. Kolmogorov \citep[c.f.][]{Nelsen:2006}: how to bound the distribution of the sum (or difference) of two random variables with fixed marginal distributions? For continuous outcomes, because $\delta = Y(1)-Y(0)$ is well-defined, our causal parameters $\tau$ and $\eta$ are determined by the distribution of the causal effect $\delta$, the difference between the treatment and control potential outcomes. Indeed, sharp bounds on the distribution of $\delta$ have been obtained by \cite{Makarov:1982}, \cite{Ruschendorf:1982}, \cite{Frank:1987} and \cite{Williamson:1990}, and recently reviewed by \cite{Fan:2014}. For ordinal outcomes however, although mathematically valid, the interpretation of $Y(1)-Y(0)$ becomes more challenging, at least in many scenarios. For example, in the context of education it is difficult to define the ``difference'' of ``Ph.D.'' and ``master.'' In behavioral research, it is unclear how to compare the improvement from ``negative'' to ``neutral'' and from ``neutral'' to ``positive.'' 

Motivated by the above, in the Supplementary Material, for ordinal outcomes we provide direct proofs of Propositions \ref{thm:tau}--\ref{thm:eta}. Our proofs directly construct the probability matrices that achieve the lower and upper bounds of $\tau$ and $\eta.$ We believe that our ``constructive'' approach helps researchers sharply bound other causal parameters (e.g., $m_j$ and $\alpha$), at least for ordinal outcomes. It is worth mentioning that, the probability matrices attaining the lower and upper bounds of $\tau$ and $\eta$ correspond to negatively associated and positively associated potential outcomes. They are both ``extreme'' scenarios. In practice, researchers may also be interested in the case with independent potential outcomes \citep{Rubin:1978, Cheng:2009, Agresti:2010, Ding:2016}, i.e.,
$
p_{kl} = p_{k+} p_{+l}
$
for all $k$ and $l.$
With independent potential outcomes, we can identify $\tau$ and $\eta$ from the marginal distributions of the potential outcomes.

\begin{proposition}\label{thm:ind}
With independent potential outcomes,
\begin{equation*}
\tau_I = \mathop{\sum\sum}_{ k \ge l}p_{k+}p_{+l},
\quad
\eta_I = \mathop{\sum\sum}_{ k > l}p_{k+}p_{+l}.
\end{equation*}
Furthermore,
$
\tau_L\le \tau_I \le \tau_U
$
and
$
\eta_L \le \eta_I \le \eta_U.
$
\end{proposition}

In cases where negatively associated potential outcomes are unlikely, we can use $\tau_I$ and $\eta_I$ as the lower bounds of $\tau$ and $\eta.$ Below we give two numerical examples to illustrate Propositions \ref{thm:tau}--\ref{thm:ind}.

\begin{example} \label{example:3}
The marginal probabilities
$
\bm p_1 = (1/5, 3/5, 1/5)^\textrm{T}
$
and
$
\bm p_0 = (2/5, 1/5, 2/5)^\textrm{T}
$
do not satisfy the stochastic dominance assumption, because $\Delta_0 = 0,$ $\Delta_1 = 1/5>0$ and $\Delta_2=-1/5<0.$ Propositions \ref{thm:tau} and \ref{thm:ind} imply that
$
\tau_L = 2/5,
$
$
\tau_I = 16/25,
$
and
$
\tau_U = 4/5.
$
The probability matrices corresponding to negatively associated, independent, and positively associated potential outcomes achieving these values are respectively
\begin{equation}
\bm P_1 =
\left(
\begin{array}{ccc}
0 & 1/5 & 0 \\
1/5 & 0 & 2/5 \\
2/5 & 0 & 0\\
\end{array}
\right),
\quad
\bm P_2 =
\left(
\begin{array}{ccc}
2/25 & 1/25 & 2/25 \\
6/25 & 3/25 & 6/25 \\
2/25 & 1/25 & 2/25 \\
\end{array}
\right),
\quad
\bm P_3 =
\left(
\begin{array}{ccc}
1/5 & 0 & 0 \\
1/5 & 1/5 & 1/5 \\
0 & 0 & 1/5 \\
\end{array}
\right).
\label{eq:joint-example3}
\end{equation}
Similarly, Propositions \ref{thm:eta} and \ref{thm:ind} imply
$
\eta_L = 1/5,$ $\eta_I = 9/25,$ and $\eta_U = 3/5.
$
\end{example}

\begin{example}\label{example:4}
The marginal probabilities
$
\bm p_1=(1/5, 1/5, 3/5)^\textrm{T}
$
and
$
\bm p_0=(3/5, 1/5, 1/5)^\textrm{T}
$
satisfy the stochastic dominance assumption, because $\Delta_0=0,$ $\Delta_1 = 2/5>0$ and $\Delta_2=2/5>0.$ Propositions \ref{thm:tau} and \ref{thm:ind} imply
$
\tau_L = 3/5,
$
$
\tau_I=22/25,
$
and
$
\tau_U = 1.
$
The probability matrices corresponding to negatively associated, independent, and positively associated potential outcomes achieving these values are respectively
\begin{equation}
\bm P_4 =
\left(
\begin{array}{ccc}
0 & 1/5 & 0 \\
0 & 0 & 1/5 \\
3/5 & 0 & 0\\
\end{array} \right),
\quad
\bm P_5 =
\left(
\begin{array}{ccc}
3/25 & 1/25 & 1/25 \\
3/25 & 1/25 & 1/25 \\
9/25 & 3/25 & 3/25 \\
\end{array} \right),
\quad
\bm P_6 =
\left(
\begin{array}{ccc}
1/5 & 0 & 0 \\
0 & 1/5 & 0 \\
2/5 & 0 & 1/5 \\
\end{array} \right).
\label{eq:joint-example4}
\end{equation}
Similarly, Propositions \ref{thm:eta} and \ref{thm:ind} imply
$
\eta_L = 2/5,$ $\eta_I = 3/5,$ and $\eta_U = 4/5.
$
\end{example}

As demonstrated in Examples \ref{example:3} and \ref{example:4}, the bounds of $\tau$ (or $\eta$) generally do not shrink to a point. However, there are some special cases in which the lower and upper bounds of $\tau$ (or $\eta$) are identical. The following corollary provides necessary and sufficient conditions for such cases.

\begin{corollary}\label{thm:bounds-shrink}
Let
$
\mathbb K = \left\{ k: p_{k+} > 0 \right\}
$
and
$
\mathbb L = \left\{ l: p_{+l} > 0 \right\}.
$
The lower and upper bounds of $\tau$ are the same, if and only if there does not exist $k_1, k_2 \in \mathbb K$ and $l_1, l_2 \in \mathbb L$ such that
\begin{equation}\label{eq:cond-bounds-shrink}
k_2 \ge l_2 > k_1 \ge l_1
\quad
\mathrm{or}
\quad
l_2 > k_2 \ge l_1 > k_1.
\end{equation}
The lower and upper bounds of $\eta$ are the same, if and only if there does not exist $k_1, k_2 \in \mathbb K$ and $l_1, l_2 \in \mathbb L$ such that
\begin{equation}\label{eq:cond-bounds-shrink-2}
l_2 \ge k_2 > l_1 \ge k_1
\quad
\mathrm{or}
\quad
k_2 > l_2 \ge k_1 > 1_1.
\end{equation}
\end{corollary}

\subsection{Covariate adjustment}\label{sec:covariate}

With pretreatment covariates, it is possible to further sharpen the bounds of the causal parameters \citep{Grilli:2008, Lee:2009, Long:2013, Mealli:2013}. Without loss of generality, we focus only on the bounds of $\tau.$ Within each level of the pretreatment covariates $\bm X = \bm x,$ 
\begin{equation*}
\tau ( \bm x )  = \mathrm{pr}   \{ Y(1) \geq Y(0)\mid \bm X = \bm x \}
\end{equation*}
is the conditional probability that the treatment is beneficial.
We can obtain the conditional lower and upper bounds
$
\tau_L ( \bm x )
$
and
$
\tau_U ( \bm x )
$
given the covariate level $\bm{x}$, then average them over the covariate distribution $F \left( \bm x \right)$, and finally obtain the adjusted bounds for $\tau:$
\begin{equation}\label{eq:covariate}
\tau_L^\prime = \int \tau_L\left( \bm x \right) F\left( d \bm x \right),
\quad
\tau_U^\prime = \int \tau_U\left( \bm x \right) F\left( d \bm x \right).
\end{equation}
\begin{proposition}\label{thm:sharper}
The adjusted bounds are tighter, i.e., $ \tau_L \le \tau_L^\prime \le \tau_U^\prime \le \tau_U.$
\end{proposition}

Proposition \ref{thm:sharper} holds intuitively, because the existence of covariates imposes more distributional restrictions on the observed data. We use the following example to illustrate Proposition \ref{thm:sharper}.

\begin{example}\label{eg:bounds-covariates}
Consider a population consisting of two sub-populations of equal sizes, labeled by a binary covariate $\bm X.$ Assume that the potential outcomes of sub-populations $\bm X=1$ and $\bm X=0$ are the independent potential outcomes in Example \ref{example:3} and \ref{example:4}. Simple algebra gives the following joint distribution, marginal distributions, and $\tau$ of the potential outcomes:
\begin{equation*}
\bm P =
\left(
\arraycolsep=2pt\def\arraystretch{1.1}
\begin{array}{ccc}
1/10 & 1/25 & 3/50 \\
9/50 & 2/25 & 7/50 \\
11/50 & 2/25 & 1/10 \\
\end{array} \right),
\quad
\bm p_1=\left( 1/5, 2/5, 2/5 \right)^\textrm{T},
\quad
\bm p_0=\left( 1/2, 1/5, 3/10 \right)^\textrm{T},
\quad
\tau = 19/25.
\end{equation*}
Without covariate information, Proposition \ref{thm:tau} implies
$
\tau_L = 1/2
$
and
$
\tau_U = 1.
$
However, if we first obtain the bounds for the two sub-populations and then average over them, we obtain sharper covariate adjusted bounds
$
\tau_L^\prime = \tau_L(1) /2 + \tau_L(0) /2 = 1/2,
$
and
$
\tau_U^\prime = \tau_U(1) /2 + \tau_U(0) /2 = 9/10.
$
\end{example}

\subsection{Identifying the bounds from observed data}

Previous subsections discussed the causal parameters $\tau$ and $\eta$ and their bounds. The causal parameters depend on the joint distribution of the potential outcomes, but the bounds depend only on the marginal distributions of the potential outcomes. In practice, the observed data provide full information about only the marginal distributions. Therefore, point estimations of the bounds can be obtained, although the causal parameters themselves are only partially identified \citep[c.f.][]{Romano:2008, Romano:2010, Richardson:2014}. 

For unit $i=1, \ldots, N,$ let the treatment indicator be $Z_i,$ and the observed outcome be
$
Y_i^\textrm{obs} = Z_i Y_i(1) + (1 - Z_i) Y_i(0).
$
To avoid conceptual complications, we consider treatment assignments that satisfy the ignorability assumption \citep{Rosenbaum:1983}, i.e., 
$
Z \ind \{ Y(1), Y(0) \} \mid \bm X.
$ 
The ignorability assumption holds by the design of randomized experiments, and cannot be validated in observational studies. Under the ignorability assumption, we define the propensity score as 
$
e(\bm X) = \mathrm{pr}(Z=1\mid \bm X),
$
which is a constant independent of $\bm X$ in completely randomized experiments. We can identify the marginal distributions of the potential outcomes by
\begin{equation*}
\mathrm{pr} \{ Y(1) = k \} = \mathrm E \left\{ \frac{Z 1(Y^\textrm{obs} = k)}{e(\bm X)} \right\},
\quad
\mathrm{pr} \{ Y(0) = l \} = \mathrm E \left\{ \frac{(1-Z) 1(Y^\textrm{obs} = l) }{ 1-e(\bm X)} \right\}. 
\end{equation*}
By replacing the expectations by their sample analogues, we obtain the moment estimators for the marginal distributions. We defer more detailed discussion about statistical inference to Section \ref{sec:inf}.

\section{Randomized Experiments with Noncompliance}\label{sec:noncompliance}

\subsection{Causal effects for compliers}

Noncompliance is an important topic in practice. For instance, in clinical trials some patients may not comply with their assigned treatments. Although noncompliance itself has been extensively investigated in the causal inference literature \citep[e.g.,][]{Angrist:1996}, there appears to be very limited discussions about causal inference of ordinal outcomes in the presence of noncompliance. To the best of our knowledge, \cite{Cheng:2009} discussed various causal parameters under the assumptions of one-sided noncompliance, and \cite{Baker:2011} generalized her results to two-sided noncompliance; both of them assumed independent potential outcomes.

Under the Stable Unit Treatment Value Assumption, for unit $i$, let $\{D_i(1), D_i(0)\}$ be the potential values of treatment received under treatment and control; the observed treatment received is therefore
$
D_i^\textrm{obs} = Z_iD_i(1)+(1-Z_i)D_i(0).
$
\citet{Angrist:1996} proposed to classify the units into four categories according to the joint values of $D_i(1)$ and $D_i(0):$ 
\begin{equation}\label{eq:principal-strata}
G_i = 
\begin{cases}
a, & \text{if } D_i(1)=1, D_i(0) = 1, \\
c, & \text{if } D_i(1)=1, D_i(0) = 0, \\
d, & \text{if } D_i(1)=0, D_i(0) = 1, \\
n, & \text{if } D_i(1)=0, D_i(0) = 0,
\end{cases}
\end{equation}
and referred to the subgroups defined in \eqref{eq:principal-strata} as always-takers ($a$), compliers ($c$), defiers ($d$) and never-takers ($n$). Let
$
\pi_g = \mathrm{pr}\left(G = g\right)
$
denote the probability of the stratum $g \in \{a, c, d, n \},$ and
\begin{equation*}
g_{kl} = \mathrm{pr} \left\{ Y(1) = k, Y(0) = l \mid G = g\right\}
\end{equation*}
be the probability of potential outcome $k$ under treatment and potential outcome $l$ under control within stratum $g.$ The $J\times J$ probability matrix $\left\{ g_{kl} \right\}_{0 \leq k,l \leq J-1}$ summarizes the joint distribution of the potential outcomes for stratum $g.$ Define
\begin{equation}\label{eq:marginal-probability}
g_{k+} = \sum_{l^\prime=0}^{J-1} g_{kl^\prime},
\quad
g_{+l} = \sum_{k^\prime =0}^{J-1} g_{k^\prime l}
\quad
(k, l = 0, 1, \ldots, J-1);
\end{equation}
the vectors $\left( g_{0+}, \ldots, g_{J-1, +} \right)^\textrm{T}$ and $\left( g_{+0}, \ldots, g_{+, J-1} \right)^\textrm{T}$ characterize the marginal distributions of the potential outcomes under treatment and control. By the law of total probability, 
\begin{equation}\label{eq:lotp}
p_{kl} = \sum_g \pi_g g_{kl},
\quad
p_{k+} = \sum_g \pi_g g_{k+},
\quad
p_{+l} = \sum_g \pi_g g_{+l}.
\end{equation}
We define the subgroup causal parameters within stratum $g$ as
\begin{equation*}
\tau_g = \mathrm{pr} \left\{ Y_i (1) \ge Y_i (0) \mid G = g \right\} = \mathop{\sum\sum}_{ k \ge  l} g_{kl},
\quad
\eta_g = \mathrm{pr} \left\{ Y_i (1) > Y_i (0) \mid G = g  \right\} = \mathop{\sum\sum}_{ k >  l} g_{kl}.
\end{equation*}

Following \citet{Angrist:1996}, we invoke the following ``standard'' assumptions: (1) complete randomization, i.e., $Z\ind \{ D(1), D(0) , Y(1), Y(0), \bm X \};$ (2) monotonicity, i.e., $D_i(1) \geq D_i(0)$ for all $i;$ (3) exclusion restriction, i.e., $D_i(1) = D_i(0)$ implies $Y_i(1) = Y_i(0).$ Monotonicity rules out the defiers with $G = d,$ and strong monotonicity further rules out the always-takers with $G = a.$ Exclusion restriction implies that $\tau_n = 1, \eta_n = 0, \tau_a = 1$ and $\eta_a = 0.$ Therefore, we discuss only the causal effects for the compliers, i.e., $\tau_c$ and $\eta_c.$

\subsection{Bounds on the causal effects for compliers}

We focus only on the monotonicity assumption, because it is more general than strong monotonicity. Under monotonicity and exclusion restriction, we identify the probabilities of always-takers, compliers and never-takers, i.e., $(\pi_a,\pi_c,\pi_n),$ and the distributions of the potential outcomes conditional on $G$ \citep{Angrist:1996, Cheng:2009, Baker:2011}, i.e., the $g_{k+}$'s and $g_{+l}$'s. Below, we establish the relationships between the causal parameters $\tau$ and $\tau_c,$ and between $\eta$ and $\eta_c.$
\begin{proposition}\label{lemma:mono-er}
$
\tau_c = \tau / \pi_c - \left( 1-\pi_c \right) / \pi_c
$
and
$
\eta_c = \eta / \pi_c.
$
\end{proposition}

Therefore, we can plug in the upper and lower bounds of $\tau$ and $\eta$ to obtain the bounds of $\tau_c$ and $\eta_c$, using the relationships in Proposition \ref{lemma:mono-er}. However, these bounds are not sharp, and the following bounds, implied by Propositions \ref{thm:tau} and \ref{thm:eta}, are narrower.

\begin{corollary}\label{coro:bounds-noncomp-general}
The sharp lower and upper bounds of $\tau_c$ are
\begin{equation*}
\tau_{c,L} = \max\limits_{0 \le j \le J-1} \left( c_{+ j} + \Delta_{c,j}\right),
\quad
\tau_{c,U} = 1 + \min\limits_{0 \le j \le J-1} \Delta_{c,j},
\end{equation*}
and the sharp lower and upper bounds of $\eta_c$ are
\begin{equation*}
\eta_{c,L} = \max\limits_{0 \le j \le J-1} \Delta_{c,j},
\quad
\eta_{c,U} = 1 + \min_{0 \le j \le J-1} \left( \Delta_{c,j} - c_{j +} \right).
\end{equation*}
\end{corollary}

Similar to Section \ref{sec:covariate}, we can use covariates to sharpen the bounds of $\tau_c.$ Within each level of the pretreatment covariates $\bm X = \bm x,$ we define the conditional probabilities that the treatment is beneficial for compliers as
\begin{equation*}
\tau_c (\bm x )  = \mathrm{pr} \{ Y(1) \geq Y(0)\mid G = c, \bm X = \bm x \},
\end{equation*}
and obtain their conditional sharp upper and lower bounds
$
\tau_{c,L} ( \bm x )
$
and
$
\tau_{c,U} ( \bm x ).
$
Because
\begin{equation*}
\tau_c = \frac{ \int \tau_c\left(\bm x\right) \pi_c\left( \bm x \right) dF\left( \bm x \right) } { \int \pi_c\left(\bm x \right) dF\left( \bm x \right) },
\end{equation*}
the bounds for $\tau_c$ become
\begin{equation}\label{eq:covariate-noncomp}
\tau^\prime_{c,L} = \frac{ \int \tau_{c, L}\left(\bm x\right) \pi_c\left( \bm x \right) dF\left( \bm x \right) } { \int \pi_c\left(\bm x \right) dF\left( \bm x \right) },
\quad
\tau^\prime_{c,U} = \frac{ \int \tau_{c, U}\left(\bm x\right) \pi_c\left( \bm x \right) dF\left( \bm x \right) } { \int \pi_c\left(\bm x \right) dF\left( \bm x \right) }.
\end{equation}
Similar to Proposition \ref{thm:sharper}, the adjusted bounds are tighter, i.e., $\tau_{c,L} \le \tau_{c,L}^\prime \le \tau_{c,U}^\prime \le \tau_{c,U}.$ 

\subsection{Using noncompliance to sharpen bounds for the whole population}

Proposition \ref{lemma:mono-er} and Corollary \ref{coro:bounds-noncomp-general} imply two new sets of bounds for $\tau$ and $\eta,$ which are tighter than those in Propositions \ref{thm:tau} and \ref{thm:eta}.
\begin{corollary}\label{corollary:bounds-tighter}
Under monotonicity and exclusion restriction, we can bound $\tau$ from below and above using
\begin{equation*}
\tau_L'' = \pi_c\tau_{c,L} + 1 - \pi_c,
\quad
\tau_U'' = \pi_c\tau_{c,U} + 1 - \pi_c,
\end{equation*}
and bound $\eta$ from below and above using
\begin{equation*}
\eta_L'' = \pi_c\eta_{c,L},
\quad
\eta_U'' = \pi_c\eta_{c,U}.
\end{equation*}
These new bounds above are narrower than those in Propositions \ref{thm:tau} and \ref{thm:eta}, because they satisfy 
$\tau_L \le \tau_L'', \tau_U = \tau_U'', \eta_L = \eta_L''$, and $\eta_U \ge \eta_U''$.
\end{corollary}

There are two reasons that we can obtain tighter bounds. First, we use the partially observed variable $G$ as a pretreatment variable. Second, the monotonicity and exclusion restriction assumptions further restrict the probability structure of the potential outcomes.

\section{Statistical Inference of the Bounds}\label{sec:inf}

\subsection{Point estimation}

In practice, we need to use the observed data to estimate the marginal probabilities of the potential outcomes and the bounds. To save space for the main text, we discuss only the bounds of $\tau$ and $\tau_c.$ We describe the point estimation procedures for the three scenarios mentioned in the previous sections -- completely randomized experiments with or without noncompliance, and unconfounded observational studies.

First, we consider completely randomized experiments without noncompliance. To estimate the unadjusted bounds, we replace $p_{k+}$ and $p_{+l}$ in Proposition \ref{thm:tau} with their sample analogues
\begin{equation*}
\widehat p_{k+} = N^{-1} \sum_{i=1}^N Z_i 1(Y_i^\textrm{obs} = k),
\quad
\widehat p_{+l} = N^{-1} \sum_{i=1}^N (1-Z_i) 1(Y_i^\textrm{obs} = l).
\end{equation*}
To estimate the covariate adjusted bounds in \eqref{eq:covariate}, we first estimate the marginal probabilities of the potential outcomes of unit $i$ given covariates $\bm x_i.$ Following \cite{Imbens:2015}, for low dimensional and discrete covariates, we can still use sample analogues. For high dimensional and continuous covariates, we can invoke parametric models such as proportional odds models. We then use the estimates, denoted as $\widehat p_{k+} \left( \bm x_i \right)$ and $\widehat p_{+l} \left( \bm x_i \right)$ respectively, to estimate the sharp lower and upper bounds of $\tau \left(\bm x_i \right),$ denoted as $\widehat \tau_L\left(\bm x_i \right)$ and $\widehat \tau_U\left(\bm x_i \right)$ respectively. Consequently, the estimated adjusted bounds of $\tau$ are
\begin{equation*}
\widehat \tau^\prime_L = N^{-1} \sum_{i=1}^N \widehat \tau_L\left(\bm x_i \right),
\quad
\widehat \tau^\prime_U = N^{-1} \sum_{i=1}^N \widehat \tau_U\left(\bm x_i \right).
\end{equation*}

Second, we consider unconfounded observational studies. If we have propensity score estimator $\widehat e (\bm x_i)$ for unit $i$, then we can estimate the marginal probabilities by
\begin{equation*}
\widehat p_{k+} = N^{-1} \sum_{i=1}^N Z_i \frac{1(Y_i^\textrm{obs} = k) } {\widehat e(X_i)},
\quad
\widehat p_{+l} = N^{-1} \sum_{i=1}^N (1-Z_i) \frac{1(Y_i^\textrm{obs} = l)}{1 - \widehat e(X_i)},
\end{equation*}
and then estimate the bounds accordingly.

Third, we consider completely randomized experiments with noncompliance. Without covariates, we use the EM algorithm \citep{Dempster:1977} to estimate $\pi_c,$ $c_{k+}$ and $c_{+l},$ and then estimate the unadjusted bounds in Corollary \ref{coro:bounds-noncomp-general}. For a more detailed description of the EM algorithm, see \cite{Baker:2011}. With covariates, we need to invoke parametric models for $G$ (e.g., multinomial logistic model given $\bm{X}$) and the marginal probabilities of the potential outcomes, and use the EM algorithm to compute the maximum likelihood of the model parameters. For more details, see the Supplementary Material, and \cite{Zhang:2009} and \cite{Frumento:2012}. After obtaining the sample analogues of $\tau_{c,L}(\bm{x}), \tau_{c,U}(\bm{x})$ and $\pi_c(\bm{x})$, we estimate the covariate adjusted bounds defined in \eqref{eq:covariate-noncomp} using a plug-in approach.

\subsection{Finite-sample bias and bias correction}
\label{sec:bias-correction}

As pointed out by several researchers \citep[e.g.,][]{Liu:1993, Manski:2000, Manski:2009}, the minimum and maximum operators in the closed-form expressions of the sharp bounds usually complicate the estimation procedure, by introducing finite-sample biases to the corresponding plug-in estimators. For example, even with unbiased estimators of the marginal probabilities (e.g., in completely randomized experiments), the estimated lower bound is positively biased. 
{\color{black} In the existing literature, this non-smoothness induced bias has been recognized and discussed by \cite{Laber:2011}, \cite{Hirano:2012} and \cite{Luedtke:2016}, under various settings.
}However, fortunately, such biases tend to diminish as the sample size increases, due to the consistency of the plug-in estimators. More importantly, as pointed out by \cite{Kreider:2007}, it is possible to effectively reduce such biases by a nonparametric bootstrap correction \citep{Parr:1983, Efron:1994}. To be more specific, let $\hat \tau_L$ denote the point estimator of $\tau_L,$ and the corresponding bias-corrected estimator is therefore $2 \hat \tau_L - \mathrm{E}_B(\hat \tau_L),$ where ``$\mathrm{E}_B$'' denotes the expectation induced by the bootstrap distribution.

The following numerical example demonstrates the magnitude of the bias associated with the plug-in estimator, and the performance of the bias-correction estimator.

\begin{example}
\label{example:bootstrap-bias-correction}
Consider a completely randomized experiment without noncompliance. To save space, we focus only on $\tau$ and its unadjusted lower bound $\tau_L$ in \eqref{eq:tau}. We choose the sample size $N \in \{ 100, 200, 500 \},$ and consider four different probability matrices. Cases 1--2 correspond to matrices $\bm P_2$ and $\bm P_3$ in \eqref{eq:joint-example3}, i.e., the independent and positively associated potential outcomes, which share the same marginal distribution but do not satisfy the stochastic dominance assumption. Cases 3--4 correspond to matrices $\bm P_5$ and $\bm P_6$ in \eqref{eq:joint-example4}, i.e., the independent and positively associated potential outcomes, which share the same marginal distribution and satisfy the stochastic dominance assumption. Columns 3--4 of Table \ref{table:demo} summarize the true values of $\tau$ and $\tau_L$ for all four cases.

For each case and fixed sample size, we independently draw $1000$ treatment assignments from a balanced completely randomized experiment. For each observed dataset, we calculate point estimates of $\tau_L,$ using the plug-in estimator and the bias-correction estimator based on $200$ bootstraps. In columns 5--6 of Table \ref{table:demo}, we report the biases of the two point estimators, from which we can draw two conclusions. First, for each case the bias of the plug-in estimator decreases as the sample size increases. Second, the bias-corrected estimator greatly reduces (in most cases by over 60\%) the bias of the plug-in estimator.

\begin{table}[ht]
\centering
\caption{Numerical examples. The first four columns contain the case label, sample size, and true values of $\tau$ and $\tau_L.$ The last two columns contain the biases of the plug-in (labeled ``p'') estimator and the bootstrap bias-corrected (labeled ``b'') estimators, calculated by 1000 repeat samplings and 200 bootstraps for each sample.}
\label{table:demo}
\begin{tabular}{cccccccccc}
\hline
Case & $N$  & $\tau$   &  $\tau_L$  & bias$_{\mathrm{p}}$  & bias$_{\mathrm{b}}$ \\
        \hline
  1 & 100 & 0.640 & 0.400 & 0.023 & 0.005 \\ 
  1 & 200 & 0.640 & 0.400 & 0.016 & 0.004 \\ 
  \vspace{.75mm}\noindent
  1 & 500 & 0.640 & 0.400 & 0.009 & 0.001 \\ 
  2 & 100 & 0.800 & 0.400 & 0.017 & -0.002 \\ 
  2 & 200 & 0.800 & 0.400 & 0.014 & 0.001 \\ 
  \vspace{.75mm}\noindent
  2 & 500 & 0.800 & 0.400 & 0.007 & -0.001 \\ 
  3 & 100 & 0.880 & 0.600 & 0.037 & 0.010 \\ 
  3 & 200 & 0.880 & 0.600 & 0.026 & 0.007 \\ 
  \vspace{.75mm}\noindent
  3 & 500 & 0.880 & 0.600 & 0.016 & 0.004 \\ 
  4 & 100 & 1.000 & 0.600 & 0.036 & 0.009 \\ 
  4 & 200 & 1.000 & 0.600 & 0.026 & 0.007 \\ 
  4 & 500 & 1.000 & 0.600 & 0.013 & 0.001 \\ 
  \hline
\end{tabular}
\end{table}
\end{example}

\subsection{Confidence intervals}
\label{sec:confidence-interval}

We discuss the construction of confidence intervals (CI) for the aforementioned causal parameters and their unadjusted or covarite adjusted bounds. For illustration, we again use $\tau$ as an example. From a practical (e.g., decision making) perspective, we aim to construct a confidence interval that covers $\tau$ at least $100(1-\alpha)$\% of the times, for pre-specified significance level $\alpha.$ Because the casual parameter is only partially identifiable, it is difficult to do so directly without additional assumptions or information. A common approach to address this challenge is to instead construct a $100(1-\alpha)$\% confidence interval for the sharp bounds $[\tau_L, \tau_U]$. Because $\tau \in [\tau_L, \tau_U],$ the results interval automatically guarantees at least $100(1-\alpha)$\% coverage rate for $\tau$ itself.

Similarly as in the point estimation procedure, because both the upper and the lower bounds involve the maximum and minimum operators, their asymptotic distributions become non-normal, rendering the construction of confidence intervals covering the bounds extremely challenging \citep{Hirano:2012}. Consequently, in practice, statisticians \citep{Cheng:2006, Yang:2016} often employed bootstrap methods \citep[e.g.,][]{Beran:1988, Beran:1990, Bickel:1997, Bickel:2008} to construct confidence intervals for partially identified parameters. Among numerous proposals, the most conceptually straightforward and transparent one is arguably the ``standard'' bootstrap procedure advocated by \cite{Horowitz:2000}, for which $1-\alpha$ confidence interval is simply $\{ \hat \tau_L - z_B(\alpha), \hat \tau_U + z_B(\alpha) \},$ where the threshold value $z_B(\alpha)$ can be obtained by solving the equation
\begin{equation*}
\mathrm{Pr}_B \{
\hat \tau_L - z_B(\alpha) \le \hat \tau_L, \hat \tau_U \le \hat \tau_U + z_B(\alpha)
\} = 1 - \alpha.
\end{equation*}
In the above equation, ``$\mathrm{Pr}_B$'' is the probability measure induced by bootstrap. Recently, several researchers \citep[e.g.,][]{Romano:2008, Romano:2010, Chernozhukov:2013} proposed more delicate methods to construct confidence intervals for partially identified parameters. Although the theoretical guarantees of the classic bootstrapped confidence intervals \citep{Horowitz:2000} are not completely established, several researchers \citep[e.g.,][]{Fan:2010, Yang:2014} have evaluated them via extensive simulation studies, and found that they achieve nominal coverage rates in general.

To empirically illustrate the validity of our inferential procedure, in Appendix C we compare \cite{Horowitz:2000}'s method to a more theoretically rigorous one, under a wide range of settings. The results suggest that, at least in our context \cite{Horowitz:2000}'s bootstrap interval performs equally well, if not slightly more ``conservative.'' Therefore, for simplicity in simulations and transparency in applications, we still use bootstrap to construct confidence intervals. We provide the code to implement the above construction approach; more sophisticated users can straightforwardly modify our code and explore more advanced methods.

\section{Simulation Studies}\label{sec:simu}

\subsection{Without noncompliance}

To save space for the main text, we focus only on $\tau$ and its sharp bounds in Proposition \ref{thm:tau}. For illustration, we first adopt the settings (i.e., sample sizes and probability matrices) in Example \ref{example:bootstrap-bias-correction}. It is worth mentioning that, for Cases 1 and 3 with independent potential outcomes, $\tau_L < \tau < \tau_U.$ For Cases 2 and 4 with positively associated potential outcomes, $\tau=\tau_U.$ In addition, by symmetry Cases 1--4 only consider $\tau > 0.5.$

For each case, we independently draw $1000$ treatment assignments from a balanced completely randomized experiment. For each observed dataset, we obtain bias-corrected estimates of $\tau_L$ and $\tau_U,$ and construct a $95\%$ confidence interval for $(\tau_L, \tau_U),$ using 200 bootstrapped samples. In Columns 5--8 of Table \ref{table:numerical}, we report the biases and standard errors of the point estimators $\widehat \tau_L$ and $\widehat \tau_U$; in Column 9, we report the coverage rates of the confidence intervals on the bounds $(\tau_L, \tau_U),$ and $\tau$ itself. We can draw several conclusions from the simulation results. First, the point estimators have small biases and standard errors. Second, the confidence intervals achieve reasonable coverage rates for the bounds $(\tau_L, \tau_U),$ although always over-cover $\tau,$ especially in cases with independent potential outcomes.

\begin{table}[htbp]
\centering
\caption{\small Simulated examples without noncompliance. The first five columns contain the case number, sample size, and true values of the parameter and its sharp lower and upper bounds. The next four columns contain the biases and standard errors of the point estimators of the bounds, and the last two columns contain the coverage properties of the confidence intervals for the bounds (labeled ``$\mathrm{coverage}_1$'') and the true parameter itself (labeled ``$\mathrm{coverage}_2$'').}
\label{table:numerical}
\begin{tabular}{ccccccccccc}
\hline
        Case & $N$  & $\tau$   &  $\tau_L$  &  $\tau_U$ & bias$_L$    &  se$_L$ &  bias$_U$ &  se$_U$ & coverage$_1$ & coverage$_2$\\
        \hline
  1 & 100 & 0.640 & 0.400 & 0.800 & 0.005 & 0.056 & 0.001 & 0.067 & 0.989 & 1.000 \\ 
  1 & 200 & 0.640 & 0.400 & 0.800 & 0.004 & 0.040 & -0.000 & 0.044 & 0.989 & 1.000 \\ 
  \vspace{.75mm}\noindent
  1 & 500 & 0.640 & 0.400 & 0.800 & 0.001 & 0.025 & -0.003 & 0.029 & 0.982 & 1.000 \\ 
  2 & 100 & 0.800 & 0.400 & 0.800 & -0.002 & 0.063 & -0.001 & 0.082 & 0.969 & 0.979 \\ 
  2 & 200 & 0.800 & 0.400 & 0.800 & 0.001 & 0.044 & -0.000 & 0.057 & 0.966 & 0.976 \\ 
  \vspace{.75mm}\noindent
  2 & 500 & 0.800 & 0.400 & 0.800 & -0.001 & 0.027 & -0.002 & 0.035 & 0.968 & 0.979 \\ 
  3 & 100 & 0.880 & 0.600 & 1.000 & 0.010 & 0.049 & 0.000 & 0.000 & 0.959 & 1.000 \\ 
  3 & 200 & 0.880 & 0.600 & 1.000 & 0.007 & 0.035 & 0.000 & 0.000 & 0.965 & 1.000 \\ 
  \vspace{.75mm}\noindent
  3 & 500 & 0.880 & 0.600 & 1.000 & 0.004 & 0.022 & 0.000 & 0.000 & 0.969 & 1.000 \\ 
  4 & 100 & 1.000 & 0.600 & 1.000 & 0.009 & 0.053 & 0.000 & 0.000 & 0.940 & 1.000 \\ 
  4 & 200 & 1.000 & 0.600 & 1.000 & 0.007 & 0.035 & 0.000 & 0.000 & 0.967 & 1.000 \\ 
  4 & 500 & 1.000 & 0.600 & 1.000 & 0.001 & 0.021 & 0.000 & 0.000 & 0.983 & 1.000 \\ 
  \hline
\end{tabular}
\end{table}

As mentioned previously, in Appendix C we conduct additional simulation studies to further examine the performance of \cite{Horowitz:2000}'s bootstrap confidence interval. The simulation results suggest that it achieves nearly nominal coverage rates for the bounds $(\tau_L, \tau_U),$ except for certain ``edge cases'' (e.g., when $\tau \approx \tau_U \approx 1$), and as expected usually over-cover $\tau.$

\subsection{With noncompliance}

To evaluate the finite sample performances of the estimators and the confidence intervals of the bounds, we conduct simulation studies under different model specifications. To save space, we focus only on the parameter $\tau_c,$ and consider six simulation cases. Cases 1--3 are indexed by the parameter  $\beta \in \{ 1,1/2,0\},$ and Cases 4--6 by $\xi \in \{ 1,1/2,0\}.$ We postpone the interpretations of $\beta$ and $\xi$ until afterwards. For each case, let the pre-treatment covariates $\bm X = \left(1, X_1, X_2 \right),$ where $X_1 \sim N(0, 1),$ and $X_2 \sim \mathrm{Bern}\left( 1/2 \right).$ For fixed $\bm X = \bm x,$ we generate the  variable $G$ from a multinomial logit model
\begin{equation*}
\pi_g \left( \bm x\right) = 
\mathrm{exp}(\bm \eta_g^\mathrm{T} \bm x)
\big/ 
\left\{ \sum_{g^\prime} \mathrm{exp}(\bm \eta_{g^\prime}^\mathrm{T} \bm x) \right\}
\quad
(g = a, c, n),
\end{equation*}
where
$
\bm \eta_c = \bm 0,
$
$
\bm \eta_a = \left( 1/2, 1, 0 \right)
$
and
$
\bm \eta_n = \left( -1/2, 1, 0 \right).
$
We generate the potential outcomes from proportional odds models.
\begin{enumerate}
\item For always-takers, let $Y_i(1) = Y_i(0),$ and their marginal distributions be
\begin{equation*}
\mathrm{logit} \left\{ \sum_{k\le j}a_{k+}\left(\bm x\right) \right\}
= \mathrm{logit} \left\{ \sum_{l\le j}a_{+l}\left(\bm x\right) \right\}
= \alpha_{a,j}  - 2x_1,
\end{equation*}
where
$
\alpha_{a,0} = -1/2
$
and
$
\alpha_{a,1} = 1.
$
\item For never-takers let $Y_i(1) = Y_i(0),$ and their marginal distributions be
\begin{equation*}
\mathrm{logit}  \left\{  \sum_{k\le j}n_{k+}\left(\bm x\right) \right\}
= \mathrm{logit}  \left\{  \sum_{l\le j}n_{+l}\left(\bm x\right) \right\}
= \alpha_{n,j},
\end{equation*}
where
$
\alpha_{a,0} = -3/2
$
and
$
\alpha_{a,1} = 0.
$
\item For compliers let $Y_i(1)$ and $Y_i(0)$ be independent, and the values of the parameters be
$
\alpha_{c,0} = -1,
$
$
\alpha_{c,1} = 1/2,
$
$
\gamma_{c,0} = 1/2
$
and
$
\gamma_{c,1} = 2.
$
\begin{enumerate}
\item For Cases 1--3, let the marginal distributions be
\begin{equation*}
\mathrm{logit} \left\{ \sum_{k\le j}c_{k+}\left(\bm x\right) \right\}
= \alpha_{c,j} - 2 \beta x_1,
\quad
\mathrm{logit} \left\{ \sum_{l\le j}c_{+l}\left(\bm x\right) \right\}
= \gamma_{c,j} + \beta x_1;
\end{equation*}
\item For Cases 4--6, let the marginals distributions be
\begin{equation*}
\mathrm{logit} \left\{ \sum_{k\le j}c_{k+}\left(\bm x\right) \right\}
= \alpha_{c,j} - 2 x_1 - \xi x_2,
\quad
\mathrm{logit} \left\{ \sum_{l\le j}c_{+l}\left(\bm x\right) \right\}
= \gamma_{c,j} + x_1 + \xi x_2.
\end{equation*}
\end{enumerate}
\end{enumerate}

For the above six cases, their true values of $\tau_c,$ unadjusted and adjusted bounds are in columns 2--4 of each sub-table of Table \ref{table:numerical-noncomp}. For Cases 1--3, the parameter $\beta$ quantifies the association between the covariates and the potential outcomes. As $\beta$ decreases, the covariate adjusted bounds become closer to the unadjusted bounds. For Cases 4--6, the parameter $\xi$ quantifies the association between the binary covariate $X_2$ and the potential outcomes of compliers.

We conduct inference without the binary covariate $X_2.$ This does not affect Cases 1--3 because $X_2$ is irrelevant in the data generating process, but does affect Cases 4--6. We purposefully design the data generating process in this way, to examine the performance of our estimators under correct and incorrect model specifications. For each case, we choose the sample size to be 1000, and independently draw $1000$ treatment assignments from a balanced completely randomized experiment. For each observed dataset, based on based on 100 bootstrapped samples, we first obtain the bias-corrected estimates of $\tau_{c,L}$ and $\tau_{c,U},$ and construct a $95\%$ confidence interval for $( \tau_{c,L}, \tau_{c,U} )$; we then estimate the bounds $\tau^\prime_{c,L}$ and $\tau^\prime_{c,U},$ and construct a $95\%$ confidence interval for $( \tau^\prime_{c,L}, \tau^\prime_{c,U} ).$

We report the simulation results in Table~\ref{table:numerical-noncomp}, in which columns 4--7 of each sub-table include the biases of the point estimators, the average lengths and coverage rates of the $95\%$ confidence intervals on the bounds. First, the point estimators of the bounds have small biases. Second, when the pretreatment covariates are associated with the potential outcomes, the confidence intervals of the bounds $( \tau_{c,L}, \tau_{c,U} )$ are longer than those of $( \tau^\prime_{c,L}, \tau^\prime_{c,U} ),$ on average. Third, the confidence intervals for the bounds $( \tau_{c,L}, \tau_{c,U} )$ and $( \tau^\prime_{c,L}, \tau^\prime_{c,U} )$ achieve reasonable coverage rates. Fourth, the performance of the bounds is robust to the missingness of the binary covariate, or, equivalently, a mis-specification of the outcome models.

\begin{table}
\caption{\small Simulated examples with noncompliance. In each sub-table, the first three columns contain the true values of the causal parameter $\tau_c$ and its lower and upper bounds, the next two columns contain the biases of the point estimators of the lower and upper bounds, and the last two columns contain the lengths and coverage rates of the $95\%$ confidence intervals for the bounds.} 
\label{table:numerical-noncomp}

\begin{subtable}{1\textwidth}
\caption{unadjusted bounds}
\centering
   \begin{tabular}{cccccccccccccc}
\hline
  Case & $\tau_c$ &  $\tau_{c,L}$  &  $\tau_{c,U}$ & bias$_L$    &  bias$_U$  & length & coverage \\
  \hline
  1 & 0.685 & 0.488 & 0.971 & -0.003 & -0.005 & 0.659 & 0.945 \\ 
  2 & 0.770 & 0.553 & 1.000 & -0.008 & 0.006 & 0.574 & 0.973 \\ 
  \vspace{.75mm}\noindent
  3 & 0.856 & 0.622 & 1.000 & 0.013 & 0.001 & 0.489 & 0.966 \\ 
  4 & 0.782 & 0.589 & 1.000 & 0.000 & 0.006 & 0.523 & 0.957 \\ 
  5 & 0.736 & 0.540 & 1.000 & -0.003 & 0.003 & 0.593 & 0.975 \\ 
  6 & 0.686 & 0.488 & 0.970 & -0.001 & -0.004 & 0.655 & 0.945 \\ 
  \hline
\end{tabular}
\end{subtable}

\bigskip
\begin{subtable}{1\textwidth}
   \caption{adjusted bounds}
\centering
   \begin{tabular}{cccccccccccccc}
\hline
  Case & $\tau_c$ & $\tau^\prime_{c,L}$  &  $\tau^\prime_{c,U}$ &  bias$_L$ &  bias$_U$ & length & coverage \\
  \hline
  1 & 0.685 & 0.503 & 0.772 & -0.001 & 0.003 & 0.466 & 0.968 \\
  2 & 0.770 & 0.563 & 0.935 & -0.006 & 0.001 & 0.530 & 0.968 \\
  \vspace{.75mm}\noindent
  3 & 0.856 & 0.622 & 1.000 & 0.001 & 0.002 & 0.489 & 0.959 \\
  4 & 0.782 & 0.602 & 0.846 & -0.002 & 0.017 & 0.436 & 0.960 \\
  5 & 0.738 & 0.556 & 0.817 & -0.001 & 0.004 & 0.447 & 0.965 \\
  6 & 0.686 & 0.503 & 0.772 & 0.008 & -0.006 & 0.466 & 0.968 \\
  \hline
\end{tabular}
\end{subtable}
\end{table}

\section{Applications}\label{sec:example}

\subsection{A taste-testing experiment without noncompliance}
We use the taste-testing experiment data in \cite{Bradley:1962} to demonstrate the estimation and inference of the proposed causal parameters. The outcome of interest $Y$ is ordinal with five categories, from ``terrible'' with $Y=0$ to ``excellent'' with $Y=4.$ We consider only three treatments C, D, E, and summarize the data and results in Table \ref{table:real}. Because negative associated potential outcomes appear unlikely in practice \citep{Ding:2016}, i.e., the three treatments are not drastically different (e.g., $Y_i(C) = 4,$ $Y_i(E) = 0$), we focus on the interpretations of the cases with independent and positive correlated potential outcomes, e.g., $\tau_I$ and $\tau_U.$ First, treatment E stochastically dominates treatment C, and the confidence intervals for $(\tau_I, \tau_U)$ and $(\eta_I, \eta_U)$ are (0.914, 1.000) and (0.651, 0.997). The results suggest that treatment E is indeed better than treatment C, because both lower confidence limits are greater than 0.5. Second, although treatment E and treatment D do not stochastically dominate each other, the confidence intervals for $(\tau_I, \tau_U)$ and $(\eta_I, \eta_U)$ are (0.656, 0.982) and (0.510, 0.886), suggesting that treatment E is better than treatment D. Therefore the proposed causal parameters $\tau$ and $\eta$ are useful for decision making, especially when the stochastic dominance assumption does not hold.

\begin{table}[ht]
\caption{Analysis of a taste-testing experiment}
\label{table:real}
\begin{subtable}{1\textwidth}
\caption{Data from \cite{Bradley:1962}}
\label{table:Bradley}
\centering
\begin{tabular}{ccccccc}
\hline
& \multicolumn{5}{c}{Outcome Categories} \\
\hline
 treatment   & 0  & 1  & 2  & 3  & 4  & row sum\\
  C          & 14 & 13 & 6  & 7  & 0  & 40    \\
  D          & 11 & 15 & 3  & 5  & 8  & 42    \\
  E          & 0  & 2  & 10 & 30 & 2  & 44    \\
  \hline
\end{tabular}
\end{subtable}

\bigskip
\begin{subtable}{1\textwidth}
\caption{Results for $\tau:$ Point estimators and confidence intervals (CIs)}
\label{table:results-tau}
\centering
   \begin{tabular}{cccccc}
\hline
            & $\widehat\tau_L$ & $\widehat\tau_I$  &  $\widehat \tau_U$  & CI for $\left(\tau_L, \tau_U \right)$ & CI for $\left(\tau_I, \tau_U \right)$ \\
            \hline
  E vs C  & 0.765 & 0.946 & 1.000 & (0.667, 1.000) & (0.914, 1.000) \\
  E vs D   & 0.630 & 0.782 & 0.856 & (0.503, 0.997) & (0.656, 0.982) \\
  \hline
\end{tabular}
\end{subtable}

\bigskip
\begin{subtable}{1\textwidth}
   \caption{Results for $\eta:$ Point estimators and confidence intervals (CIs)}
   \label{table:results-eta}
\centering
   \begin{tabular}{cccccc}
\hline
           & $\widehat \eta_L$ & $\widehat \eta_I$  &  $\widehat \eta_U$  & CI for $\left(\eta_L, \eta_U \right)$  & CI for $\left(\eta_I, \eta_U \right)$ \\
            \hline
  E vs C  & 0.623 & 0.780 & 0.870 & (0.480, 1.000) & (0.651, 0.997) \\
  E vs D  & 0.573 & 0.659 & 0.738 & (0.413, 0.896) & (0.510, 0.886) \\
  \hline
\end{tabular}
\end{subtable}
\end{table}

\subsection{A sexual assault education program without noncompliance}
Between September 2011 and February 2013, three universities in Canada (Windsor, Guelph and Calgary) conducted the Sexual Assault Resistance Education (SARE) Trial. The SARE trial investigates whether the enhanced Assess, Acknowledge and Act (AAA) program, which consist of numerous activities (e.g., lectures, discussions and practices), can help prevent sexual assaults. 451 first-year female students from the above universities where randomly assigned to the treatment group ($Z = 1$) with access to AAA, and 442 were randomly assigned to the control group ($Z = 0$) with brochures containing general information on sexual assault. The primary outcome $Y$ is ordinal with six categories, from ``complete rape'' with $Y=0$ to ``no reporting of any non-consensual sexual contact'' with $Y=5.$

We summarize the data and results in Table \ref{table:sare}. Because both the treatment and control groups receives useful information on sexual assault prevention, negatively associated potential outcomes seem unlikely. Therefore, we again focus on independent and positively correlated potential outcomes. The confidence intervals for $(\tau_I, \tau_U)$ and $(\eta_I, \eta_U)$ are (0.758, 1.000) and (0.554, 0.999), suggesting that AAA is indeed beneficial, because both lower confidence limits are greater than 0.5. Our findings corroborate the recommendations by \cite{Senn:2015}.

\begin{table}[ht]
\caption{Analysis of the SARE trial}
\label{table:sare}
\begin{subtable}{1\textwidth}
\caption{Data from \cite{Senn:2015}}
\label{table:senn}
\centering
\begin{tabular}{cccccccc}
\hline
& \multicolumn{5}{c}{Outcome Categories} \\
\hline
                       & 0  & 1  & 2  & 3  & 4  & 5 & row sum         \\
  treatment      & 23 & 15 & 48  & 67  & 121 & 177 & 451    \\
  control          & 42 & 40 & 62  & 103  & 184  & 11 &  442   \\
  \hline
\end{tabular}
\end{subtable}

\bigskip
\begin{subtable}{1\textwidth}
   \caption{Results for $\tau$ and $\eta:$ Point estimators and confidence intervals (CIs)}
\centering
   \begin{tabular}{cccccc}
\hline
           & Lower bound & Indep.  &  Upper bound  & CI for (L, U)  & CI for (I, U)  \\
            \hline
 $\tau$ & 0.636 & 0.783 & 1.000 & (0.598, 1.000) & (0.758, 1.000) \\ 
  $\eta$ & 0.368 & 0.604 & 0.962 & (0.311, 1.000) & (0.554, 0.999) \\ 
  \hline
\end{tabular}
\end{subtable}
\end{table}

\subsection{A job training program with noncompliance}

In the mid-1990s, Mathematica Policy Research conducted an experiment that randomly enrolled eligible applicants into the Job Corps program \citep{Schochet:2003, Lee:2009}. We re-analyzed the dataset from 1995 with 13499 units. For more detailed descriptions of the dataset, see \cite{Zhang:2009} and \cite{Frumento:2012}. In the following analysis, $Z = 1$ if an applicant was enrolled in the program, and $Z=0$ otherwise; $D=1$ if an applicant actually participated in the program, and $D=0$ otherwise. The strong monotonicity assumption with $D_i(0) = 0$ for all $i$ holds by design. Using the hourly wage after 52 weeks of enrollment, we create a three-level ordinal outcome $Y$ as follows: $Y=0$ for zero wage because of unemployment, $Y=1$ for low wage (no more than 4.25 U.S dollars, 150 \% of the minimal wage at the time the data was collected), and $Y=2$ for high wage (more than 4.25 U.S dollars). In the following analysis we take into account covariates such as gender, age, education, and marital status.

We report the results in Table \ref{table:real-2}. Similar as before, we focus on independent and positively correlated potential outcomes. For both causal parameters $\tau_c$ and $\eta_c,$ the confidence intervals for the lower and upper bounds become narrower when we take covariates into account. Similarly as the previous example, we focus on the interpretations of the cases with independent and positive correlated potential outcomes. The confidence intervals with or without covariates for $(\tau_I, \tau_U)$ suggest that the hourly wages of more than 70\% of participants does not decrease because of the job training program. Additionally, the confidence intervals with or without covariates for $(\eta_I, \eta_U)$ suggest that the hourly wages of roughly 20\%--30\% of participants strictly increase because of the job training program.

\begin{table}[ht]
\caption{Analysis of the Job Corps Program}
\label{table:real-2}
\begin{subtable}{1\textwidth}
\caption{Results for $\tau$: point estimators and confidence intervals (CIs)}
\centering
   \begin{tabular}{cccccc}
\hline
            & $\widehat\tau_{c,L}$ & $\widehat\tau_{c,I}$  &  $\widehat \tau_{c,U}$  & CI for $\left(\tau_{c,L}, \tau_{c,U} \right)$ & CI for $\left(\tau_{c,I}, \tau_{c,U} \right)$ \\
            \hline
  w/o Covariates & 0.561 & 0.707 & 0.912 & (0.536, 0.938) & (0.687, 0.934) \\ 
  w/ Covariates & 0.592 & 0.723 & 0.912 & (0.570, 0.932) & (0.700, 0.932) \\ 
  \hline
\end{tabular}
\end{subtable}

\bigskip
\begin{subtable}{1\textwidth}
   \caption{Results for $\eta$: point estimators and confidence intervals (CIs)}
\centering
   \begin{tabular}{cccccc}
\hline
           & $\widehat \eta_{c,L}$ & $\widehat \eta_{c,I}$  &  $\widehat \eta_{c,U}$  & CI for $\left(\eta_{c,L}, \eta_{c,U} \right)$ & CI for $\left(\eta_{c,I}, \eta_{c,U} \right)$  \\
            \hline
  w/o Covariates & 0.005 & 0.209 & 0.351 & (0.000, 0.362) & (0.199, 0.363) \\ 
  w/ Covariates & 0.004 & 0.193 & 0.320 & (0.000, 0.331) & (0.180, 0.331) \\ 
  \hline
\end{tabular}
\end{subtable}
\end{table}

As a final note, we use this example to illustrate Corollary \ref{corollary:bounds-tighter}. First, without the noncompliance information, the estimators of the bounds of $\tau$ are 
$
\widehat \tau_L = 0.558
$
and
$
\widehat \tau_U = 0.937,
$
with $95\%$ confidence interval (0.541, 0.954); the estimators of the bounds of $\eta$ are
$
\widehat \eta_L = 0.004
$
and
$
\widehat \eta_U = 0.379,
$
with $95\%$ confidence interval (0.000, 0.388). With the noncompliance information, the estimators of the bounds of $\tau$ are 
$
\widehat \tau_L'' = 0.683
$
and
$
\widehat \tau_U'' = 0.937,
$
with $95\%$ confidence interval (0.666, 0.954); the estimator of the bounds of $\eta$ are
$
\widehat \eta_L'' = 0.004
$
and
$
\widehat \eta_U'' = 0.254,
$
with $95\%$ confidence interval (0.000, 0.262). Therefore, the noncompliance information in return improves the inference of $\tau$ and $\eta$ for the whole population.

\section{Concluding Remarks}\label{sec:discuss}

We proposed to use two causal parameters to evaluate treatment effect on ordinal outcomes, and derived the explicit forms of their sharp bounds by using only the marginal distributions of the potential outcomes. Although we advocate the use of parameters $\tau$ and $\eta$ to measure treatment effects, we acknowledge that some other causal parameters may also provide information in practice \citep[e.g.][]{Agresti:2010, Volfovsky:2015}. For general parameters, although deriving the explicit forms of the bounds may be difficult, we can use numerical methods. For instance, for another widely-used parameter, the relative treatment effect $\alpha = \tau + \eta -1$ \citep{Agresti:2010}, we can use numerical linear programs to calculate its maximum and minimum values under the constraints in (4).

\section*{Acknowledgments}

The authors thank Drs Avi Feller at Berkeley and Luke W. Miratrix at Harvard for various helpful suggestions. 
{\color{black} Jiannan Lu is grateful to several colleagues at Microsoft Corporation, especially Dr. Alex Deng, for continuous encouragement and support. Peng Ding gratefully acknowledges financial support from the Institute for Education Science (grant No. R305D150040) and the National Science Foundation (DMS grant No. 1713152).
Thoughtful comments from the co-Editor, Dr. Dan McCaffrey, and three anonymous reviewers, have helped improve the quality and presentation of our paper significantly.}

\bibliographystyle{apalike}
\bibliography{ordinal_bounds}

\newpage
\appendix

\section{Proofs of Lemma, Propositions and Corollaries}\label{sec:proof}

\subsection{A lemma and its proof}\label{subsec:proof-lemmas}

We first state a lemma extending a result in \cite{Strassen:1965}. This lemma plays a central role in our later proofs, and is also of independent interest. 

\begin{lemma}\label{LEMMA:1}
Assume that $\left( x_0, \ldots, x_{n-1} \right)$ and $\left( y_0, \ldots, y_{n-1} \right)$ are nonnegative constants.
\begin{enumerate}[label= (\alph*), ref = \ref{LEMMA:1}(\alph*)]

\item \label{LEMMA:1-a}
If
$
\sum_{r=s}^{n-1} x_r \ge \sum_{r=s}^{n-1} y_r
$
for all $s = 0, \ldots, n-1,$ there exists an $n\times n$ lower triangular matrix
$
\bm A_n = (a_{kl})_{0 \le k,l \le n-1}
$
with nonnegative elements such that
\begin{equation}\label{eq:lemma1-1}
\sum_{l^\prime=0}^{n-1} a_{kl^\prime} \le x_k,
\quad
\sum_{k^\prime=0}^{n-1} a_{k^\prime l}= y_l
\quad
(k, l = 0, \ldots, n-1).
\end{equation}

\item \label{LEMMA:1-b}
If
$
\sum_{r=s}^{n-1} x_r \le \sum_{r=s}^{n-1} y_r
$
for all $s = 0, \ldots, n-1,$
there exists an $n\times n$ upper triangular matrix
$
\bm B_n = (b_{kl})_{0 \le k,l \le n-1}
$
with nonnegative elements such that
\begin{equation}
\label{eq:lemma1-2}
\sum_{l^\prime = 0}^{n - 1} b_{kl^\prime}= x_k,
\quad
\sum_{k^\prime = 0}^{n - 1} b_{k^\prime l} \le y_l
\quad
(k, l=0, \ldots, n-1).
\end{equation}

\item \label{LEMMA:1-c}
If
$
\sum_{r=0}^s x_r \le \sum_{r=0}^s y_r
$
for all $s = 0, \ldots, n-1,$
there exists an $n\times n$ lower triangular matrix
$
\bm C_n = (p_{kl})_{0 \le k,l \le n-1}
$
with nonnegative elements such that
\begin{equation}
\label{eq:lemma1-3}
\sum_{l^\prime=0}^{n-1} p_{k l^\prime}= x_k,
\quad
\sum_{k^\prime=0}^{n-1} p_{k^\prime l} \le y_l
\quad
(k, l = 0, \ldots, n-1).
\end{equation}

\item \label{LEMMA:1-d}
If
$
\sum_{r=0}^s x_r \ge \sum_{r=0}^s y_r
$
for all $s = 0, \ldots, n-1,$
there exists an $n\times n$ upper triangular matrix
$
\bm D_n = (d_{kl})_{0 \le k,l \le n-1}
$
with nonnegative elements such that
\begin{equation}
\label{eq:lemma1-4}
\sum_{l^\prime = 0}^{n-1} d_{k l^\prime} \le x_k,
\quad
\sum_{k^\prime = 0}^{n-1} d_{k^\prime l}= y_l
\quad
(k, l = 0, \ldots, n-1).
\end{equation}

\item \label{LEMMA:1-e}
If we further assume
$
\sum_{r=0}^{n-1} y_r = \sum_{r=0}^{n-1} x_r,
$
the above inequalities in (\ref{eq:lemma1-1})--(\ref{eq:lemma1-4}) all reduce to equalities, i.e., the matrices $\bm A_n,$ $\bm B_n,$ $\bm C_n$ and $\bm D_n$ have $\left( x_0, \ldots, x_{n-1} \right)$ and $\left( y_0, \ldots, y_{n-1} \right)$ as their row and column sums.
\end{enumerate}
\end{lemma}

Note that if all the $x_i$'s are the same as the $y_i$'s, we can simply construct a diagonal matrix with elements $x_i$'s or $y_i$'s. The following proof deals with general cases.

\begin{proof}[Proof of Lemma 1(a)]
We prove by induction. When $n=1,$ we let $\bm A_1 = y_0\geq 0,$ and Lemma 1(a) holds because $y_0 \le x_0.$ When $n \ge 2,$ suppose Lemma 1(a) holds for $n-1.$ In particular, for any $\left( x_1, \ldots, x_{n-1} \right)$ and $\left( y_1, \ldots, y_{n-1} \right)$ such that
$
\sum_{r=s}^{n-1} x_r \ge \sum_{r=s}^{n-1} y_r
$
for all
$
s = 1, \ldots, n-1,
$
there exists a lower triangular matrix
$
\bm A_{n-1} = (a_{kl})_{1 \le k,l \le n-1}
$
with nonnegative elements
such that
\begin{equation}\label{eq:tl.n-1}
\sum_{l^\prime=1}^{n-1} a_{kl^\prime} \le x_k,
\quad
\sum_{k^\prime=1}^{n-1} a_{k^\prime l}= y_l
\quad
(k, l = 1, \ldots, n-1).
\end{equation}

To prove that Lemma 1(a) holds for $n,$ we let
\begin{equation*}
\bm A_n =
\left(\begin{array}{cc}
a_{00} &  \bm 0^\textrm{T} \\
\bm a & \bm A_{n-1} \\
\end{array}\right),
\end{equation*}
where $a_{00}$ and $\bm a = (a_{10}, \ldots, a_{n-1, 0})^\textrm{T}$ are defined for two separate cases below.

\begin{enumerate}[label=(\arabic*)]

\item
$y_0 <  x_0.$ We let $a_{00}=y_0,$ and $a_{k0}=0$ for all $k=1, \ldots, n-1.$ Clearly, $\bm A_n$ has nonnegative elements, and satisfies the row and column sum conditions in Lemma 1(a) holds;

\item
$y_0 \geq x_0.$ We let $a_{00}=x_0,$ and
\begin{equation}\label{eq:1stcol}
a_{k0}=
\left(y_0 - a_{00}\right) \frac{x_k - \sum_{l^\prime=1}^{n-1} a_{kl^\prime}}{\sum_{k^\prime=1}^{n-1}\left(x_{k^\prime} - \sum_{l^\prime=1}^{n-1} a_{k^\prime l^\prime}\right)}  \geq 0
\quad
(k=1, \ldots, n-1).
\end{equation}
This construction guarantees that the column sums of $\bm A_n$ are $y_l$'s. Furthermore, because $\bm A_{n-1}$ satisfies \eqref{eq:tl.n-1}, we have
\begin{eqnarray}
\label{eq:lemma1a}
\sum_{k^\prime=1}^{n-1}\left( x_{k^\prime} - \sum_{l^\prime=1}^{n-1} a_{k^\prime l^\prime} \right) & = & \sum_{k^\prime=1}^{n-1} x_{k^\prime} - \sum_{k^\prime=1}^{n-1}\sum_{l^\prime=1}^{n-1} a_{k^\prime l^\prime} =
\sum_{k^\prime=1}^{n-1} x_{k^\prime} - \sum_{l^\prime=1}^{n-1}\sum_{k^\prime=1}^{n-1} a_{k^\prime l^\prime} \nonumber \\
& = & \sum_{k^\prime=1}^{n-1} x_{k^\prime} - \sum_{k^\prime=1}^{n-1} y_{k^\prime} \ge y_0 - x_0 = y_0 - a_{00} > 0.
\end{eqnarray}
Formulas \eqref{eq:1stcol} and \eqref{eq:lemma1a} imply that
$
a_{k0} \le x_k - \sum_{l^\prime=1}^{n-1} a_{kl^\prime}
$
and therefore $\sum_{l^\prime=0}^{n-1} a_{kl^\prime} \le x_k$
for
$
k = 1, \ldots, n - 1.
$
\end{enumerate}
Therefore Lemma 1(a) holds for $n,$ and the proof is complete.
\end{proof}

\begin{proof}[Proof of Lemma 1(b).]
By applying Lemma 1(a) to
$
\left( y_0, \ldots, y_{n-1} \right)$ and $\left( x_0, \ldots, x_{n-1} \right),
$
we obtain a lower triangular matrix
$
\widetilde{\bm{B_n}} = (\tilde b_{kl})_{0 \le k,l \le n-1}
$
with nonnegative elements
such that
\begin{equation*}
\sum_{k^\prime=0}^{n-1} \tilde{b}_{k^\prime l}= x_k,
\quad
\sum_{l^\prime=0}^{n-1} \tilde{b}_{kl^\prime} \le y_k
\quad
(k, l = 0, \ldots, n-1).
\end{equation*}
Let $\bm B_n = \widetilde{\bm{B_n}}^\textrm{T},$ and the proof is complete.
\end{proof}

\begin{proof}[Proof of Lemma 1(c).]
By applying Lemma 1(a) to
$
\left( y_{n-1}, \ldots, y_0 \right)$ and $\left( x_{n-1}, \ldots, x_0 \right),
$
we obtain a lower triangular matrix
$
\widetilde{\bm{C_n}} = (\tilde p_{kl})_{0 \le k,l \le n-1}
$
with nonnegative elements
such that
\begin{equation*}
\sum_{k^\prime=0}^{n-1} \tilde{c}_{k^\prime l}= x_{n-l-1},
\quad
\sum_{l^\prime=0}^{n-1} \tilde{c}_{kl^\prime} \le y_{n-k-1}
\quad
(k, l = 0, \ldots, n-1).
\end{equation*}
Let
$
\bm C_n = \left( \tilde{c}_{n-l-1, n-k-1} \right)_{0 \le k,l \le n-1},
$
and the proof is complete.
\end{proof}

\begin{proof}[Proof of Lemma 1(d).]
By applying Lemma 1(c) to
$
\left( y_0, \ldots, y_{n-1} \right)$ and $\left( x_0, \ldots, x_{n-1} \right),
$
we obtain a lower triangular matrix
$
\widetilde{\bm{D_n}} = (\tilde d_{kl})_{0 \le k,l \le n-1}
$
with nonnegative elements
such that
\begin{equation*}
\sum_{l^\prime=0}^k \tilde{d}_{k l^\prime} = y_k,
\quad
\sum_{k^\prime=l}^{n-1} \tilde{d}_{k^\prime l} \le x_k
\quad
(k, l = 0, \ldots, n-1).
\end{equation*}
Let $\bm D_n = \widetilde{\bm{D_n}}^\textrm{T},$ and the proof is complete.
\end{proof}

\begin{proof}[Proof of Lemma 1(e).]
In addition to the proof of Lemma 1(a), we further need to show that if
$
\sum_{r=0}^{n-1} y_r = \sum_{r=0}^{n-1} x_r,
$
the row sums of the constructed matrix $\bm A_n$ are $x_k$'s. In the induction of the proof of Lemma 1(a), if we have constructed matrix $A_{n-1},$ the case with $y_0<x_0$ would not happen.
We consider only the case with $y_0 \geq x_0.$ Because the lower triangular matrix $A_{n-1}$ has the column sums $y_l$'s, and
$
\sum_{r=0}^{n-1} y_r = \sum_{r=0}^{n-1} x_r,
$
we have
$$
\sum_{k^\prime=1}^{n-1}\left( x_{k^\prime} - \sum_{l^\prime=1}^{n-1} a_{k^\prime l^\prime} \right) = \sum_{k^\prime=1}^{n-1} x_{k^\prime} - \sum_{k^\prime=1}^{n-1} y_{k^\prime}  =  y_0 - x_0 = y_0 - a_{00} > 0.
$$
The above formula, coupled with the construction of the first column of $\bm A_n$ in \eqref{eq:1stcol}, gives $a_{k0} = x_k - \sum_{l'=1}^{n-1} a_{kl'}$ and thus $\sum_{l'=0}^{n-1} a_{kl'} = x_k$ for all $k.$
\end{proof}

\subsection{Proof of Proposition 1}

Now we prove the main Proposition 1, and the proofs for other propositions and corollaries are relatively straightforward.

\begin{proof}[Proof of Proposition 1.]
For all $j=0, 1, \ldots, J-1,$ 
\begin{eqnarray}
\tau &=& \mathop{\sum\sum}_{ k \geq   l}  p_{kl}  =    1 -   \mathop{\sum\sum}_{ k <  l}  p_{kl}
\nonumber \\
&\leq& 1 - \sum_{k<j} \sum_{ l\ge j}p_{kl}
 = 1 - \left(   \sum_{k=0}^{J-1} \sum_{l\geq j} p_{kl} - \sum_{k\geq j} \sum_{l \geq j} p_{kl} \right)
\label{eq:upper-ineq-1} \\
&\leq& 1 - \left(   \sum_{k=0}^{J-1} \sum_{l\geq j} p_{kl} - \sum_{k\geq j} \sum_{l=1}^{J-1} p_{kl} \right)
=1 - \left(\sum_{l\geq j}  p_{+l}    - \sum_{k\geq j}  p_{k+}\right)
\label{eq:upper-ineq-2} \\
& = & 1 + \Delta_j, \nonumber
\end{eqnarray}
and
\begin{eqnarray}
\tau & = & \mathop{\sum\sum}_{k\geq l} p_{kl}  \nonumber  \\
& \ge & \sum_{k\geq j } \sum_{l\leq j} p_{kl}
=  \sum_{k\geq j } \sum_{l = 0}^{J-1} p_{kl}  - \sum_{k\geq j } \sum_{l >j} p_{kl} \label{eq:lower-ineq-1} \\
& \ge &   \sum_{k\geq j } \sum_{l = 0}^{J-1} p_{kl}  - \sum_{k=0}^{J-1} \sum_{l >j} p_{kl} = \sum_{k\geq j } p_{k+} - \sum_{l >j} p_{+l} \label{eq:lower-ineq-2} \\
& = & p_{+j} + \Delta_j, \nonumber
\end{eqnarray}
which implies that
$
\tau_L \leq \tau \leq \tau_U.
$

We now construct two probability matrices attaining the lower and upper bounds respectively, using Lemma \ref{LEMMA:1}.

We first construct a probability matrix attaining the upper bound $\tau_U.$ Let
$$
j_1     = \min\left\{ 0\leq j ' \leq J-1 :  \Delta_{j'} =  \min_{0\leq j\leq J-1} \Delta_j  \right\}
$$
be the minimum index $j$ that attains the minimum value of $\Delta_j$'s. To attain $\tau_U,$ the equalities in \eqref{eq:upper-ineq-1} and \eqref{eq:upper-ineq-2} must hold, i.e.,
\begin{equation}\label{eq:guide}
\mathop{\sum\sum}_{ k <  l}  p_{kl}  = \sum_{k<j_1} \sum_{ l\ge j_1}p_{kl}  ,
\quad
 \sum_{k\geq j_1} \sum_{l \geq j_1} p_{kl} = \sum_{k\geq j_1} \sum_{l=1}^{J-1} p_{kl} .
\end{equation}

If $j_1=0,$ $\min_{0\leq j\leq J-1}  \Delta_j = \Delta_0 = 0,$ implying that $\Delta_j = \sum_{k=j}^{J-1} p_{k+} - \sum_{l=j}^{J-1} p_{+l}  \geq 0$ for all $j,$ i.e., the marginal probabilities satisfy the stochastic dominance assumption. According to Lemma \ref{LEMMA:1-e}, there exists a lower triangular probability matrix $\bm P$ with marginal probabilities $\bm p_1 = \left(p_{0+}, \ldots, p_{J-1,+}\right)^\textrm{T}$ and $\bm p_0 = \left( p_{+0}, \ldots, p_{+, J-1}\right)^\textrm{T}.$ Correspondingly, $\tau = 1 + \Delta_0 = 1.$

If $j_1>0,$ the constraints in (\ref{eq:guide}) force some elements of the probability matrix to be zeros. To be more specific, the constraints in (\ref{eq:guide}) imply that the probability matrix has the following block structure:
\begin{equation}\label{eq:structure-p-tauu}
\bm P =
\left(
    \begin{array}{cc}
   \bm P_{\textrm{tl}}  & \bm P_{\textrm{tr}} \\
   \bm 0 & \bm P_{\textrm{br}} \\
  \end{array}\right),
\end{equation}
where the $j_1\times j_1$ sub-matrix
$
\bm P_{\textrm{tl}}
$
on top left and the $(J-j_1)\times (J-j_1)$ sub-matrix
$
\bm P_{\textrm{br}}
$
on bottom right are both lower triangular, and the $j_1\times (J-j_1)$ sub-matrix
$
\bm P_{\textrm{tr}}
$
on top right has no restrictions.

Because $\Delta_{j_1} \le \Delta_j$ for all $j=0,1,\ldots, J-1,$ we have
$$
\sum_{k=j}^{j_1-1} p_{k+} \ge \sum_{l=j}^{j_1-1} p_{+l}
\quad
(j=0, \ldots, j_1-1);
\quad
\sum_{k=j_1}^{j} p_{k+} \le \sum_{l=j_1}^{j} p_{+l}
\quad
(j = j_1,\ldots,J-1).
$$
Given the above two sets of constraints on the marginal probabilities, we construct the probability matrix $\bm P$ in three steps.
\begin{enumerate}[label=(\arabic*)]

\item
We apply Lemma \ref{LEMMA:1-a} to $\left( p_{0+}, \ldots, p_{j_1-1, +} \right)$ and $\left( p_{+0}, \ldots, p_{+, j_1-1} \right),$ and obtain a lower triangular matrix
$
\bm P_{\textrm{tl}}=\left( p_{kl} \right)_{0 \le k,l \le j_1-1}
$
with nonnegative elements
such that
$$
\sum_{l^\prime = 0}^{j_1-1} p_{k l^\prime} \le p_{k+},
\quad
\sum_{k^\prime = 0}^{j_1-1} p_{k^\prime l}= p_{+l}
\quad
(k, l = 0, \ldots, j_1 - 1).
$$
\item
We apply Lemma \ref{LEMMA:1-c} to $\left(p_{j_1+}, \ldots, p_{J-1, +} \right)$ and $\left(p_{+j_1}, \ldots, p_{+, J-1} \right),$ and obtain a lower triangular matrix
$
\bm P_{\textrm{br}}=\left( p_{kl} \right)_{j_1 \le k,l \le J-1}
$
with nonnegative elements
such that
$$
\sum_{l^\prime = j_1}^{J-1} p_{k l^\prime}= p_{k+},
\quad
\sum_{k^\prime = j_1}^{J-1} p_{k^\prime l} \le p_{+l}
\quad
(k,l = j_1, \ldots, J-1).
$$
\item
We construct
$
\bm P_{\textrm{tr}}=\left( p_{kl} \right)_{0 \le k \le j_1 - 1, j_1 \le l \le J-1}
$
by letting
$$
p_{kl}= \left(p_{k+} - \sum_{l^\prime = 0}^{j_1 - 1} p_{k l^\prime}\right)
        \left(p_{+l} - \sum_{k^\prime = j_1}^{J-1} p_{k^\prime l}\right) \geq 0
\quad
(k = 0,\ldots,j_1-1; \; l = j_1,\ldots,J-1).
$$
\end{enumerate}
The constructed probability matrix $\bm P$ has marginal probabilities $\bm p_1=\left(p_{0+}, \ldots, p_{J-1,+}\right)^\textrm{T}$ and $\bm p_0=\left( p_{+0}, \ldots, p_{+, J-1}\right)^\textrm{T}.$ What is more, by \eqref{eq:structure-p-tauu} the $\tau$ of $\bm P$ is the sum of all the elements in $\bm P_{\textrm{tl}}$ and $P_{\textrm{br}},$ which we construct in the above (1) and (2). Therefore, we have
$$
\tau = \sum_{l^\prime = 0}^{j_1-1} p_{+l^\prime} + \sum_{k^\prime = j_1}^{J-1} p_{k^\prime +} = 1 + \Delta_{j_1},
$$
which implies that the probability matrix $\bm P$ attains $\tau_U.$

We then construct a probability matrix attaining the lower bound in $\tau_L.$ Let
$$
j_2 = \min\left\{ j' :   p_{+j'} + \Delta_{j'}  = \max_{0\leq j \leq J-1} (p_{+j} + \Delta_j)  \right\}
$$
be the minimum index $j$ that attains the maximum value of $(p_{+j} + \Delta_j)$'s. To attain $\tau_L,$ the equalities in \eqref{eq:lower-ineq-1} and \eqref{eq:lower-ineq-2} must hold, i.e.,
\begin{equation}\label{eq:guide2}
 \mathop{\sum\sum}_{k\geq l} p_{kl}  = \sum_{k\geq j_2 } \sum_{l\leq j_2 } p_{kl}  ,
\quad
 \sum_{k \ge j_2} \sum_{l >j_2} p_{kl} =  \sum_{k=0}^{J-1} \sum_{l >j_2} p_{kl}.
\end{equation}

If $j_2=0,$ from \eqref{eq:guide2} we know that the elements in the lower triangular part but not in the first column of the probability matrix $\bm P$ are all zeros, i.e.,
\begin{equation}\label{eq:structure-p-taul-0}
\bm P =
\left(
\begin{array}{cc}
  \bm p & \bm P_{\textrm{tr}} \\
   p_{J-1, 0} & 0^\textrm{T}\\
\end{array}
\right),
\end{equation}
where $\bm p=\left( p_{0, 0}, \ldots, p_{J - 2, 0} \right)^\textrm{T},$ and the $(J-1)\times (J-1)$ sub-matrix $P_{\mathrm{tr}}$ on top right is upper triangular. Because $p_{+0} + \Delta_0 \ge p_{+j} + \Delta_j$ for all $j,$ we have
$$
\sum_{k = 0}^j p_{k+} \ge \sum_{l = 0}^j p_{+, l+1}
\quad
(j = 0, \ldots, J - 2).
$$
Applying Lemma \ref{LEMMA:1-d} to $\left( p_{0+}, \ldots, p_{J - 2, +} \right)$ and $\left(p_{+1}, \ldots, p_{+, J - 1} \right),$ we obtain an upper triangular matrix
$
\bm P_{\textrm{tr}}=\left( p_{kl} \right)_{0 \le k \le J-2, 1 \le l \le J-1}
$
with nonnegative elements
such that
$$
\sum_{l^\prime = 1}^{J - 1} p_{k l^\prime} \le p_{k+},
\quad
\sum_{k^\prime = 0}^{J - 2} p_{k^\prime l}= p_{+l}
\quad
(k=0, \ldots, J - 2; \; l = 1, \ldots, J - 1).
$$
To complete the construction, let $p_{J-1, 0} = p_{J-1, +},$ and
$$
p_{k0} = p_{k+} - \sum_{l^\prime = 1}^{J - 1} p_{k l^\prime} \geq 0
\quad
(k = 0, \ldots, J - 2).
$$
The constructed probability matrix $\bm P$ has marginal probabilities $\bm p_1 = \left(p_{0+}, \ldots, p_{J-1,+}\right)^\textrm{T}$ and $\bm p_0 = \left( p_{+0}, \ldots, p_{+, J-1}\right)^\textrm{T}.$ Moreover, by \eqref{eq:structure-p-taul-0} the $\tau$ of $\bm P$ is the sum of all the elements in the first column. Therefore
$
\tau = p_{+0} = p_{+0} + \Delta_0,
$
which implies that $\bm P$ attains $\tau_L.$

If $j_2=J-1,$ the proof is similar to the above case with $j_2=0.$ If $0 < j_2 < J-1,$ because the first equality in \eqref{eq:guide2} is equivalent to
$$
\sum_{k< j_2 } \sum_{l\le k } p_{kl} + \sum_{k \ge j_2 } \sum_{l\le k } p_{kl} = \sum_{k \ge j_2 } \sum_{l\le j_2 } p_{kl},
$$
the probability matrix $\bm P$ must satisfy the following constraints:
\begin{enumerate}[label = (C\arabic*)]

\item \label{condition-1}
 For all $k=0, \ldots, j_2-1,$ $p_{kl}=0$ for all $l=0, \ldots, k.$

\item \label{condition-2}
 For all $k=j_2+1, \ldots, J-1,$ $p_{kl}=0$ for all $l=j_2+1, \ldots, k.$

\end{enumerate}
Similarly, because the second equality in \eqref{eq:guide2} is equivalent to
$$
\sum_{k \ge j_2 } \sum_{l > j_2} p_{kl} = \sum_{k \ge j_2 } \sum_{l > j_2 } p_{kl} + \sum_{k < j_2 } \sum_{l > j_2 } p_{kl},
$$
the probability matrix $\bm P$ must further satisfy the following constraint:
\begin{enumerate}[label = (C\arabic*)]
\setcounter{enumi}{2}
\item \label{condition-3}
$p_{kl}=0,$ for all $k=0, \ldots, j_2-1$ and $l=j_2 + 1, \ldots, J-1.$
\end{enumerate}
The constraints in \ref{condition-1}, \ref{condition-2} and \ref{condition-3} imply that $\bm P$ must have the following block structure:
\begin{equation}\label{eq:structure-p-taul}
\bm P =
\left(
\begin{array}{cc}
  \left(\bm 0, \bm P_\textrm{tl} \right)  & \bm 0 \\
  \bm P_\textrm{bl} & \left( \begin{array}{c} \bm P_\textrm{br} \\ \bm 0^\textrm{T} \end{array} \right) \\
\end{array}
\right)
\end{equation}
where the $j_2\times j_2$ sub-matrix
$
\bm P_{\textrm{tl}}
$
and the $(J-j_1-1)\times (J-j_1-1)$ sub-matrix
$
\bm P_{\textrm{br}}
$
are both upper triangular, and the $(J - j_2)\times (j_2 + 1)$ sub-matrix
$
\bm P_{\textrm{bl}}
$
on bottom left has no restrictions.

Because $p_{+j_2} + \Delta_{j_2} \ge p_{+j} + \Delta_j$ for all $j,$ we have
$$
\sum_{k = j}^{j_2 - 1} p_{k+} \le \sum_{l = j}^{j_2 - 1}  p_{+, l+1}
\quad
(j = 0, \ldots, j_2 - 1);
\quad
\sum_{k = j_2}^s p_{k+} \ge \sum_{l = j_2}^s p_{+, l+1}
\quad
(j = j_2, \ldots, J - 2).
$$
Given the above two sets of constraints for the marginal probabilities, we construct the probability matrix $\bm P$ in three steps.
\begin{enumerate}[label=(\arabic*)]

\item
We apply Lemma \ref{LEMMA:1-b} to $\left( p_{0+}, \ldots, p_{j_2 - 1, +} \right)$ and $\left( p_{+1}, \ldots, p_{+, j_2} \right),$ and obtain an upper triangular matrix
$
\bm P_{\textrm{tl}}=\left( p_{kl} \right)_{0 \le k \le j_2-1, 1 \le l \le j_2}
$
with nonnegative elements
such that
$$
\sum_{l^\prime = 1}^{j_2} p_{k l^\prime}= p_{k+},
\quad
\sum_{k^\prime = 0}^{j_2 - 1} p_{k^\prime l} \le p_{+l}
\quad
(k = 0, \ldots, j_2 - 1; \; l = 1, \ldots, j_2).
$$

\item
We apply Lemma \ref{LEMMA:1-d} to $\left( p_{j_2+}, \ldots, p_{J - 2, +} \right)$ and $\left( p_{+, j_2 + 1}, \ldots, p_{+, J - 1} \right),$ and obtain an upper triangular matrix
$
\bm P_{\textrm{br}}=\left( p_{kl} \right)_{j_2 \le k \le J-2, j_2+1 \le l \le J-1}
$
with nonnegative elements
such that
$$
\sum_{l^\prime = j_2 + 1}^{J - 1} p_{k l^\prime} \le p_{k+}
,
\quad
\sum_{k^\prime = j_2}^{J - 2} p_{k^\prime l}= p_{+l}
\quad
(k=j_2, \ldots, J - 2; \; l = j_2 + 1, \ldots, J - 1).
$$

\item
We construct
$
\bm P_{\textrm{bl}}=\left( p_{kl} \right)_{j_2 \le k \le J-1, 0 \le l \le j_2}
$
by letting
$$
p_{kl}= \left( p_{k+} - \sum_{l^\prime = j_2 + 1}^{J-1} p_{k l^\prime}\right) \left(p_{+l} - \sum_{k^\prime = 0}^{j_2 - 1}p_{k^\prime l}\right) \geq 0
\quad
(k = j_2, \ldots, J - 1; \; l = 0, \ldots, j_2).
$$
\end{enumerate}
The constructed probability matrix $\bm P$ has marginal probabilities $\bm p_1 = \left(p_{0+}, \ldots, p_{J-1,+}\right)^\textrm{T}$ and $\bm p_0 = \left( p_{+0}, \ldots, p_{+, J-1}\right)^\textrm{T}.$ Moreover, by \eqref{eq:structure-p-taul} the corresponding $\tau$ is the sum of all the elements in $\bm P_{\textrm{bl}},$ which we construct in the above (3). Therefore,
\begin{eqnarray*}
\tau  = 1 - \sum_{k^\prime = 0}^{j_2-1} p_{k^\prime + } - \sum_{l^\prime = j_2 + 1}^{J-1} p_{+l^\prime}
= \sum_{k^\prime = j_2}^{J-1} p_{k^\prime + } - \sum_{l^\prime = j_2 + 1}^{J-1} p_{+l^\prime} = p_{+j_2} + \Delta_{j_2},
\end{eqnarray*}
which implies that $\bm P$ attains $\tau_L.$
\end{proof}

\subsection{Proofs of other propositions}\label{sec:proof-other-propositions}

\begin{proof}[Proof of Proposition 2]
Because
$
\eta = 1- \mathrm{pr}\left\{Y_i (0) \ge Y_i(1)\right\},
$
its lower bound is one minus the upper bound of $\mathrm{pr}\left\{Y_i (0)\ge Y_i(1)\right\}.$
By switching the treatment and control labels, we can bound $\mathrm{pr}\left\{Y_i (0) \ge Y_i (1)\right\}$ from the above by
$$
\mathrm{pr}\left\{Y_i (0)\ge Y_i (1)\right\} \leq
1 - \max_{0\leq j\leq J-1} \Delta_j,
$$
which implies that
$
\eta_L = \max_{0\leq j\leq J-1} \Delta_j .
$

Similarly, the upper bound of $\eta$ equals one minus the lower bound of $\mathrm{pr}\left\{Y_i (0)\ge Y_i (1)\right\}.$ By switching the treatment and control labels, we can bound $\mathrm{pr}\left\{Y_i (0) \ge Y_i (1)\right\}$ from below by
$$
\mathrm{pr}\left\{Y_i (0)\ge Y_i (1)\right\}  \geq
\max_{0\leq j\leq J-1}  \left( p_{j+} - \Delta_j \right),
$$
which implies that
$
\eta_U = 1 + \min_{0\leq j\leq J-1} \left(\Delta_j - p_{j +}\right).
$
\end{proof}

\begin{proof}[Proof of Proposition 3]
With independent potential outcomes, the probability matrix $P$ has elements
$
p_{kl} = p_{k+} p_{+l}
$
for
$
k
$
and
$
l.
$
We obtain $\tau_I$ and $\eta_I$ by their definitions. Obviously, they are between their lower and upper bounds, i.e.,
$
\tau_L\le \tau_I \le \tau_U
$
and
$
\eta_L \le \eta_I \le \eta_U.
$
\end{proof}

\begin{proof}[Proof of Proposition 4]
The proof follows Lee (2009). Because any value of $\tau$ within the covariate adjusted bounds $[ \tau_L^\prime, \tau_U^\prime ]$ must be compatible with the distributions of $\left\{ Y(1), \bm X\right\}$ and $\left\{Y(0), \bm X \right\},$ it must also be compatible with the distributions of $Y(1)$ and $Y(0)$ by discarding $\bm X.$ Therefore, any value of $\tau$ within the adjusted bounds $[ \tau_L^\prime, \tau_U^\prime ]$ must also be within the unadjusted bounds $[\tau_L, \tau_U].$ Consequently, the adjusted bounds are tighter, i.e., $[ \tau_L^\prime, \tau_U^\prime ] \subset [\tau_L, \tau_U].$ Similar arguments apply to the covariate adjusted bounds and the unadjusted bounds for $\tau_c.$
\end{proof}

\begin{proof}[Proof of Proposition 5]
Under monotonicity, by the law of total probability, we have
\begin{equation*}
\tau = \pi_c\tau_c + \pi_a \tau_a + \pi_n \tau_n.
\end{equation*}
Under exclusion restriction, we have $\tau_a = 1$ and $\tau_n = 1,$ yielding
\begin{equation*}
\tau  = \pi_c \tau_c + 1 - \pi_c,
\end{equation*}
which implies that
\begin{equation*}
\tau_c = \tau/\pi_c- \left( 1-\pi_c \right) / \pi_c.
\end{equation*}
Analogously, we have
$
\eta = \pi_c \eta_c,
$
which implies that
$
\eta_c = \eta/\pi_c.
$
\end{proof}

\subsection{Proofs of the corollaries}\label{subsec:proof-corollaries}

\begin{proof}[Proof of Corollary 1]
By Proposition 1, $\tau=1$ if and only if
$
\min_{0\le j \le J-1}\Delta_j=0.
$
Because $\Delta_0=0,$ this is equivalent to
$
\Delta_j \ge 0
$
for all $j$, i.e., the stochastic dominance assumption holds.
\end{proof}

\begin{proof}[Proof of Corollary 2] 
Similar to the proof of Proposition 2, because
$
\eta = 1- \mathrm{pr}\left\{Y_i (0) \ge Y_i(1)\right\},
$
(9) immediately implies (10). Therefore, we need only to prove that (9) is the sufficient and necessary condition that the lower and upper bounds of $\tau$ are the same, i.e., $\tau_L = \tau_U.$

First we prove the necessity of the condition. Assume that it does not hold, i.e., there does exist $k_1, k_2 \in \mathbb K$ and $l_1, l_2 \in \mathbb L$ such that (9) holds. In this case we construct two probability matrices with the same marginal probabilities but different values of $\tau.$ The first probability matrix is 
$
\bm P =  \left( p_{k+} p_{+l} \right)_{0\leq k,l\leq J-1}.
$
For the second probability matrix, let 
$\xi = \min 
\left( 
p_{k_1, +} p_{+, l_1}, 
\;
p_{k_2, +} p_{+, l_2} 
\right),
$
which is a positive constant. We then apply the following matrix operation to the $2\times 2$ sub-matrix of the first probability matrix:
\begin{equation*}
\begin{pmatrix}
p_{k_1 l_1} & p_{k_1 l_2} \\
p_{k_2 l_1} & p_{k_2 l_2}
\end{pmatrix}
\longrightarrow
\begin{pmatrix}
p_{k_1 l_1} - \xi & p_{k_1 l_2} + \xi \\
p_{k_2 l_1} + \xi & p_{k_2 l_2} - \xi
\end{pmatrix}
\end{equation*} 
The above operation preserves the marginal probabilities, and the difference of $\tau$ between the first and second probability matrices is $\xi$, if $ k_2 \ge l_2 > k_1 \ge l_1$, and $-\xi$, if $l_2 > k_2 \ge l_1 > k_1$.

Second, we prove the sufficiency of the condition. If $|\mathbb K| = 1$ or $|\mathbb L|=1,$ the probability matrix degenerates and consequently we have $\tau_L = \tau_{c,U}.$ If $|\mathbb K| \ge 2$ and $|\mathbb L| \ge 2,$ let 
$
k_* = \min_{k\in \mathbb K} k
$
and
$
k^* = \max_{k\in \mathbb K} k
$
be the minimal and maximal indices of nonzero $p_{k+}$'s, and
$
l_* = \min_{l\in \mathbb L} l
$
and
$
l^* = \max_{l\in \mathbb L} l
$
the minimal and maximal indices of nonzero $p_{+l}$'s. A useful fact that we repeatedly use is that if $p_{k+} = 0,$ then $p_{kl} = 0$ for all $l.$ Similarly, if $p_{+l} = 0,$ then $p_{kl} = 0$ for all $k.$

Because $k_*, k^*$ and $l_*, l^*$ cannot satisfy (9), we discuss the two following cases based on the relative locations of the two intervals $[k_*, k^*]$ and $[l_*, l^*]:$
\begin{enumerate}
\item ``Non-overlapping,'' i.e., $k_* \ge l^*$ or $k^* < l_*:$

\begin{enumerate}

\item If
$
k_* \ge l^*,
$
we prove that $p_{kl} = 0$ for all $k<l.$ Assume the claim does not hold, then there exists $k^\prime < l^\prime $ such that $p_{k^\prime l^\prime }>0,$ then $p_{k^\prime+}>0$ and $p_{+l^\prime}>0. $ This implies that $k_* \le k^\prime < l^\prime \le l^*,$ contradicting the initial assumption. Therefore, $\tau_L = \tau_U = 1.$

\item If
$
k^* < l_*,
$
similarly $p_{kl} = 0$ for all $k\ge l,$ implying that $\tau_L = \tau_U = 0.$

\end{enumerate}

\item ``Inclusive,'' i.e., $l^* > k^* > k_* \ge l_*$ or $k^* \ge l^* > l_* > k_*:$ 

\begin{enumerate}

\item If $l^* > k^* > k_* \ge l_*,$ and furthermore if there exists $l^\prime \in \mathbb L$ such that $k_* < l^\prime \le k^*,$ then $l^\prime \ne l_*$ and $l^\prime \ne l^*.$ Moreover, $k_*, k^*$ and $l^\prime, l^*$ satisfy (9), contradicting the initial assumption. Therefore for all $l \in \mathbb L,$ $l \le k_*$ or $l > k^*.$ Consequently, 
\begin{equation*}
\begin{split}
\tau & = \sum_{k\ge l} p_{kl} 1  ( k \in \mathbb K, l \in \mathbb L ) = \sum_{k\ge l} p_{kl} 1 ( k \in \mathbb K, l \in \mathbb L, l \le k_* ) \\
& = \sum p_{kl} 1 ( k \in \mathbb K, l \in \mathbb L, l \le k_* ) = \sum_{l \le k_*, l \in \mathbb L} p_{+l}
\end{split}
\end{equation*}
is identifiable, which implies that $\tau_L = \tau_U.$

\item If $k^* \ge l^* > l_* > k_*,$ similarly as above for all $k \in \mathbb K,$ $k < l_*$ or $k \ge l^*.$ Consequently, 
\begin{equation*}
\tau 
= \sum_{k\ge l} p_{kl} 1 ( k \in \mathbb K, l \in \mathbb L ) \\
= \sum_{k\ge l} p_{kl} 1 ( k \in \mathbb K, l \in \mathbb L, k \ge l^* ) \\
= \sum_{k \ge l^*, k \in \mathbb K} p_{k+}
\end{equation*}
is identifiable, which implies that $\tau_L = \tau_U.$
\end{enumerate}
\end{enumerate}
\end{proof}

\begin{proof}[Proof of Corollary 3]
The proof follows directly from Propositions 1 and 2.
\end{proof}

\begin{proof}[Proof of Corollary 4]
The closed-form expressions for
$
\tau_{c,L}'',
$
$
\tau_{c,U}'',
$
$
\eta_{c,L}''
$
and
$
\eta_{c,U}''
$
follow directly from Proposition 5 and Corollary 2. Furthermore, under the monotonicity and exclusion restriction assumptions, we have
\begin{equation*}
\Delta_j = \pi_c \Delta_{c,j}
\quad
(j=0, \ldots, J-1).
\end{equation*}
Therefore, for the upper bound of $\tau$, we have
\begin{equation*}
\tau_U = 1 - \pi_c + \pi_c (1 + \min \Delta_{c,j} )
= \tau_U'',
\end{equation*}
and for the lower bound, we have
\begin{equation*}
\tau_L \le  \max ( p_{+j} - 1 + \pi_c + \pi_c\Delta_{c,j} )
= 1 - \pi_c + \pi_c \max ( c_{+j} + \Delta_{c,j} )
= \tau_L''.
\end{equation*}
The first step holds because under the strong monotonicity assumption
$
n_{+j}\pi_n \le \pi_n,
$
and under the monotonicity assumption
$
a_{+j}\pi_a + n_{+j}\pi_n \le \pi_a + \pi_n.
$
Similar arguments apply to the bounds of $\eta.$
\end{proof}

\section{Details of the EM Algorithm with Noncompliance}
\label{sec::em}

Let
$
\bm X_i = \bm x_i,
$
$
Z_i = z_i,
$
$
D_i^\mathrm{obs} = d_i
$
and
$
Y_i^\mathrm{obs} = y_i
$
be the values of the pretreatment covariates, treatment assigned, treatment received and observed outcome of the $i$th unit. We write the likelihood function as
\begin{eqnarray*}
L\left(\bm\theta\right) & = & \prod_{i:z_i=1,d_i=1} \left\{\pi_a\left(\bm x_i\right) a_{y_i, +}\left(\bm x_i\right) + \pi_c\left(\bm x_i\right) c_{y_i, +}\left(\bm x_i\right)\right\}
\quad
\times \prod_{i:z_i=1,d_i=0} \left\{\pi_n\left( \bm x_i\right) n_{y_i, +}\left(\bm x_i\right)\right\} \\
& \times & \prod_{i:z_i=0,d_i=1}\left[\pi_a\left(\bm x_i\right) a_{+, y_i}\left(\bm x_i\right)\right]
\quad
\times  \prod_{i:z_i=0,d_i=0}\left\{\pi_n\left(\bm x_i\right) n_{+, y_i}\left(n; \bm x_i\right) + \pi_c\left(\bm x_i\right) c_{+, y_i}\left(\bm x_i\right)\right\}.
\end{eqnarray*}

Let $G_i = g_i$ be the value of the principal stratification variable of the $i$th unit. By treating it as missing data, we write the complete-data log-likelihood as:
\begin{eqnarray*}
l_\textrm{C}\left(\bm\theta\right) & = & \sum_{i:g_i = a, z_i=1} \left[\log \left\{\pi_a\left(\bm x_i\right)\right\} + \log \left\{a_{y_i, +}\left(\bm x_i\right)\right\} \right]
\quad
+  \sum_{i:g_i=c, z_i=1} \left[\log \left\{\pi_c\left(\bm x_i\right)\right\} + \log \left\{c_{y_i, +}\left(\bm x_i\right)\right\} \right] \nonumber \\
& + & \sum_{i:g_i=n, z_i=1} \left[\log \left\{\pi_n\left(\bm x_i\right)\right\} + \log \left\{n_{y_i, +}\left(\bm x_i\right)\right\} \right]
\quad
+  \sum_{i:g_i=a, z_i=0} \left[\log \left\{\pi_a\left(\bm x_i\right)\right\} + \log \left\{a_{+, y_i}\left(\bm x_i\right)\right\} \right] \nonumber \\
& + & \sum_{i:g_i=c, z_i=0} \left[\log \left\{\pi_c\left(\bm x_i\right)\right\} + \log \left\{c_{+, y_i}\left(\bm x_i\right)\right\} \right]
\quad
+ \sum_{i:g_i=n, z_i=0} \left[\log \left\{\pi_n\left(\bm x_i\right)\right\} + \log \left\{n_{+, y_i}\left(\bm x_i\right)\right\} \right],
\end{eqnarray*}

We denote the realizations of (4) and (5) for the $i$th unit when evaluated at the true parameter value $\bm \theta$ as
$
\pi_g\left(\bm x_i\right),
$
$
g_{k+}\left(\bm x_i\right)
$
and
$
g_{+l}\left(\bm x_i\right),
$
and those when evaluated at the $t$th iteration of the parameter estimate $\bm \theta^{(t)}$ as
$
\pi_g^{(t)}\left(\bm x_i\right),
$
$
g_{k+}^{(t)}\left(\bm x_i\right)
$
and
$
g_{+l}^{(t)}\left(\bm x_i\right).
$
The EM algorithm proceeds as follows. Given the current $t$th iteration of the parameter estimate $\bm \theta^{(t)},$ we obtain the updated $(t+1)$th iteration $\bm \theta^{(t+1)}$ as follows:
\begin{enumerate}
\item E-Step: obtain the conditional expectation of the complete-data log-likelihood, given observed data and the current parameter estimate $\bm \theta^{(t)},$ by finding the (current) conditional probabilities of the principal stratum $g,$ denoted as $\pi^{(t)}_{g,i}.$

\begin{enumerate}

\item For all $i$ such that $z_i=1$ and $d_i=1,$ let
$
\pi^{(t)}_{n,i} = 0,
$
and
$$
\pi^{(t)}_{g,i} = \frac{\pi_g^{(t)}\left(\bm x_i\right) g_{y_i, +}^{(t)}\left(\bm x_i\right)}{\pi_a^{(t)}\left(\bm x_i\right) a_{y_i, +}^{(t)}\left(\bm x_i\right) + \pi_c^{(t)}\left(\bm x_i\right) c_{y_i, +}^{(t)}\left(\bm x_i\right)}
\quad
(g = a, c).
$$

\item For all $i$ such that $z_i=1$ and $d_i=0,$ let
$
\pi_{a,i}^{(t)} = 0,
$
$
\pi_{c,i}^{(t)} = 0,
$
and
$
\pi_{n,i}^{(t)} = 1.
$

\item For all $i$ such that $z_i=0$ and $d_i=1,$ let
$
\pi_{a,i}^{(t)} = 1,
$
$
\pi_{c,i}^{(t)} = 0,
$
and
$
\pi_{n,i}^{(t)} = 0.
$

\item For all $i$ such that $z_i=0$ and $d_i=0,$ let
$
\pi^{(t)}_{a,i} = 0
$
and
$$
\pi^{(t)}_{g,i} = \frac{\pi_g^{(t)}\left(\bm x_i\right) g_{+, y_i}^{(t)}\left( \bm x_i\right)}{\pi_c^{(t)}\left( \bm x_i\right) c_{+, y_i}^{(t)}\left( \bm x_i\right) + \pi_n^{(t)}\left( \bm x_i\right) n_{+, y_i}^{(t)}\left( \bm x_i\right)}
\quad
(g = c, n).
$$

\end{enumerate}

\item M-Step: obtain the updated parameter estimate $\bm \theta^{(t+1)},$ by maximizing the conditional expectation with respect to $\bm \theta.$ To do this, we adopt the following two-step procedure:

\begin{enumerate}

\item Obtain $\bm \theta_\mathrm{PS}^{(t+1)},$ the updated estimates of the parameters in the model for the principal strata, by maximizing the following objective function:
\begin{eqnarray*}
\mathit{F} \left(\bm \theta_\textrm{PS}\right) & = &
\sum_{i:z_i=1,d_i=1} \left\{ \pi_{a,i}^{(t)} \log \pi_a\left(\bm x_i\right) + \pi_{c,i}^{(t)} \log \pi_c\left( \bm x_i\right) \right\} \quad
+ \sum_{i:z_i=1,d_i=0} \pi_{n,i}^{(t)} \log \pi_n\left( \bm x_i\right) \\
& + & \sum_{i:z_i=0,d_i=1} \pi_{a,i}^{(t)} \log \pi_a\left( \bm x_i\right)
\quad
+ \sum_{i:z_i=0,d_i=0} \left\{ \pi_{c,i}^{(t)} \log \pi_c\left( \bm x_i\right) + \pi_{n,i}^{(t)} \log \pi_n\left(\bm x_i\right)\right\}.
\end{eqnarray*}    
The optimization problem is equivalent to fitting the following weighted multinomial logistic regression:

\begin{enumerate}

\item For $i$ such that $z_i=1$ and $d_i=1,$ create two new observations for the regression: one always-taker with weight $\pi_{a,i}^{(t)}$ and one complier with weight $\pi_{c,i}^{(t)}.$

\item For $i$ such that $z_i=0$ and $d_i=0,$ create two new observations: one complier with weight $\pi_{c,i}^{(t)}$ and one never-taker with weight $\pi_{n,i}^{(t)}.$

\item For $i$ such that $z_i=1$ and $d_i=0,$ create one never-taker with weight 1.

\item For $i$ such that $z_i=0$ and $d_i=1,$ create one always-taker with weight 1.

\end{enumerate}

\item Similarly, obtain $\bm \theta_\mathrm{PO}^{(t+1)},$ the updated estimates of the parameters in the model for the potential outcomes, by fitting weighted proportional odds models:

\begin{enumerate}

\item For $g=a,$ use all $i$ such that $z_i=1$ and $d_i=1$ with weight $\pi_{a,i}^{(t)},$ and all $i$ such that $z_i=0$ and $d_i=1$ with weight 1.

\item For $g=c,$ use all $i$ such that $z_i=1$ and $d_i=1$ and all $i$ such that $z_i=0$ and $d_i=0$ with weight $\pi_{c,i}^{(t)}.$

\item For $g=n,$ use all $i$ such that $z_i=1$ and $d_i=0$ with weight 1, and all $i$ such that $z_i=0$ and $d_i=0$ with weight $\pi_{n,i}^{(t)}.$

\end{enumerate}

\end{enumerate}

\end{enumerate}

\section{Additional Simulation Studies}\label{sec:additional-simu}

\subsection{The alternative method}

We further examine the performances of \cite{Horowitz:2000}'s bootstrap method to construct confidence intervals for partially identified parameters, by comparing it to a more rigorous approach proposed by \cite{Jiang:2018} as follows:
\begin{enumerate}
    
    \item Denote $p_{+j} + \Delta_j$ and $1 + \Delta_j,$ the ``building blocks'' of $\tau_L,$ as $L_j$ and $U_j$ respectively. Obtain their finite-sample estimates $\hat L_j$ and $\hat U_j,$ and let
    $$
    \hat q = \mathrm{argmax}_{0 \le j \le J-1} \hat L_j,
    \quad
    \hat r = \mathrm{argmin}_{0 \le j \le J-1} \hat U_j
    $$
    be the minimum indices attaining the maximum value of $\hat L_j$'s, and the minimum value of $\hat U_j$'s, respectively;
    
    \item Estimate the standard deviations of $\hat L_{\hat q}$ and $\hat U_{\hat r},$ via standard bootstrap. Denote the resulted standard errors as $\hat \sigma_{\hat q}$ and $\hat \sigma_{\hat r},$ respectively;
    
    \item Correspondingly, the $(1-\alpha)$ confidence interval for $\tau$ is
    $
    ( 
    \hat L_{\hat q} - C \hat \sigma_{\hat q},
    \hat U_{\hat r} + C \hat \sigma_{\hat r}
    ),
    $
    where we the threshold value $C$ by solving the equation
    $$
    \Phi\left\{
    C + \frac{\hat U_{\hat r} - \hat L_{\hat q}}{\max (\hat \sigma_{\hat q}, \hat \sigma_{\hat r})}
    \right\} 
    - \Phi(-C) = 1 - \alpha.
    $$
    
\end{enumerate}
By utilizing the fact that $\hat L_j$'s and $\hat U_j$'s are jointly asymptotically normal, \cite{Jiang:2018} generalized previous results by \cite{Imbens:2004} and proved that the resulted confidence interval achieves nominal coverage rate for $\tau.$

\subsection{Simulation results}

We choose the sample size $N = 200,$ and generate 50 simulation cases (i.e., probability matrices) by repeating the following procedure 50 times:
\begin{enumerate}
    \item Let $U_r \stackrel{i.i.d}{\sim} \mathrm{Unif}(0, 1)$ for all $r = 1, \ldots, 6;$
    \item Let the marginal probabilities 
    $$
    \bm p_1 = (U_1, U_2, U_3) / \sum_{k=1}^3 U_k,
    \quad
    \bm p_0 = (U_4, U_5, U_6) / \sum_{l=4}^6 U_l;
    $$
    \item Based on $\bm p_1$ and $\bm p_0$ construct a probability matrix corresponding to positively correlated potential outcomes (i.e., $\tau = \tau_U$).
\end{enumerate}
It is worth mentioning that, we intentionally use uniform random variables for the marginal probabilities, to ``objectively'' explore a wide range of potential outcome distributions. Additionally, we choose the case with positively correlated potential outcomes, because previous simulations suggest that it appear to be the most challenging.

We only focus on the coverage properties of the confidence intervals by \cite{Horowitz:2000} and \cite{Jiang:2018}, respectively. Following the main text, for each of the 100 probability matrices, we independently draw $1000$ treatment assignments from a balanced completely randomized experiment, and calculate two interval estimates of the bounds $(\tau_L, \tau_U).$ In Figure \ref{fig:add-simu} we report the coverage rates of the two confidence intervals, for both the bounds  $(\tau_L, \tau_U)$ and the parameter $\tau$ itself. We can draw several conclusions from the simulation results. First, both intervals achieve nominal coverage rates for the bounds $(\tau_L, \tau_U),$ expect for some ``edge cases'' (i.e., when $\tau \approx \tau_U \approx 1$). Second, both intervals (inevitably) over-covers $\tau.$ Third, for all cases \cite{Horowitz:2000}'s bootstrapped interval performs equally well compared to \cite{Jiang:2018}'s interval, if not better. For example, for Case 16 where
$$
\bm p_1 = (0.484, 0.134, 0.382),
\quad
\bm p_0 = (0.504, 0.295, 0.201),
\quad
\tau_L = 0.504,
\quad
\tau = \tau_U = 1,
$$
\cite{Horowitz:2000}'s and \cite{Jiang:2018}'s intervals achieve 0.957 and 0.896 coverage rates for $(\tau_L, \tau_U)$ respectively, and 0.978 and 0.914 for $\tau$ respectively.

\begin{figure}[H]
\centering
\begin{subfigure}{.5\textwidth}
  \centering
  \includegraphics[width = 1\linewidth]{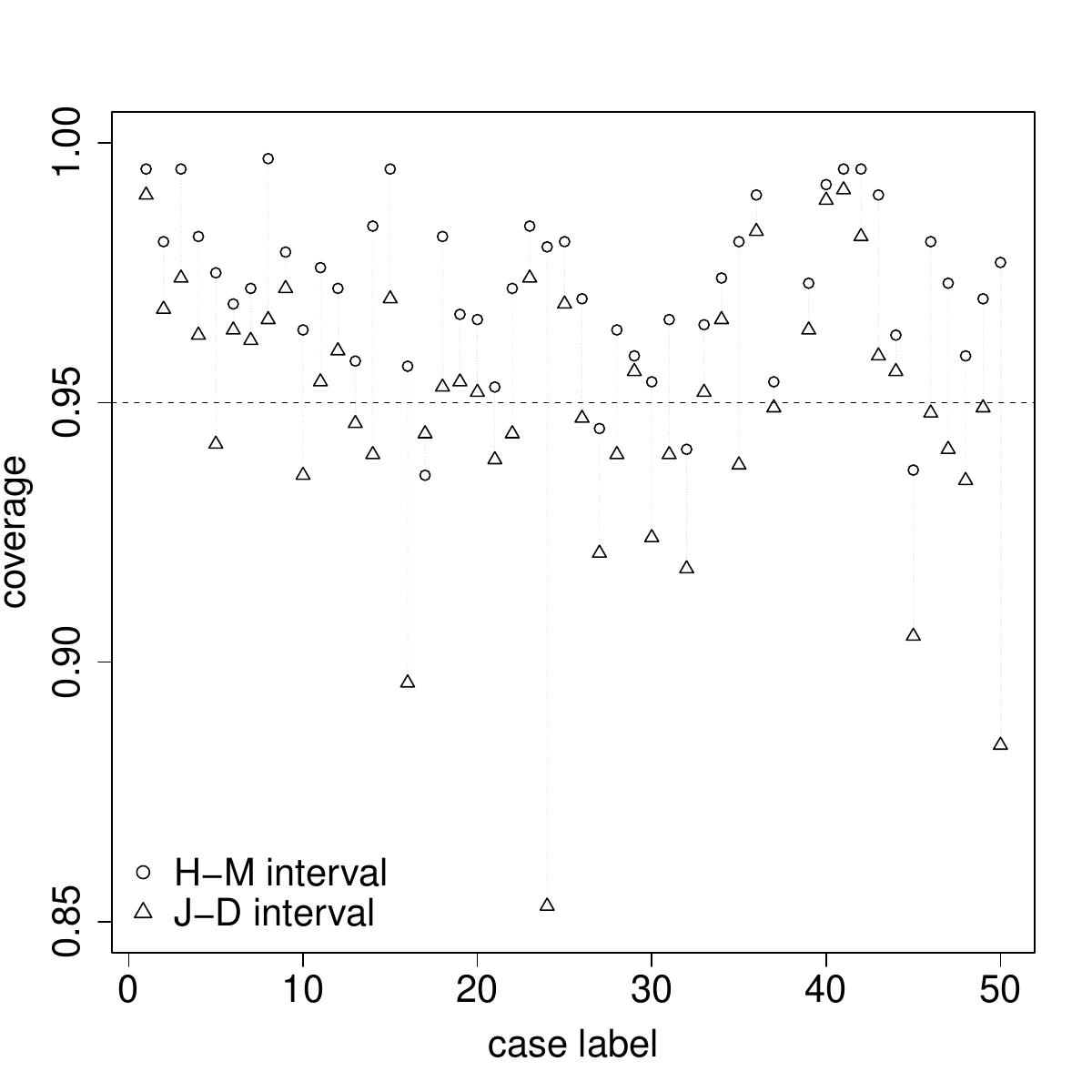}
  \caption{\small Coverage rates for the bounds $(\tau_L, \tau_U).$}
  \label{fig:sub1}
\end{subfigure}%
\begin{subfigure}{.5\textwidth}
  \centering
  \includegraphics[width = 1\linewidth]{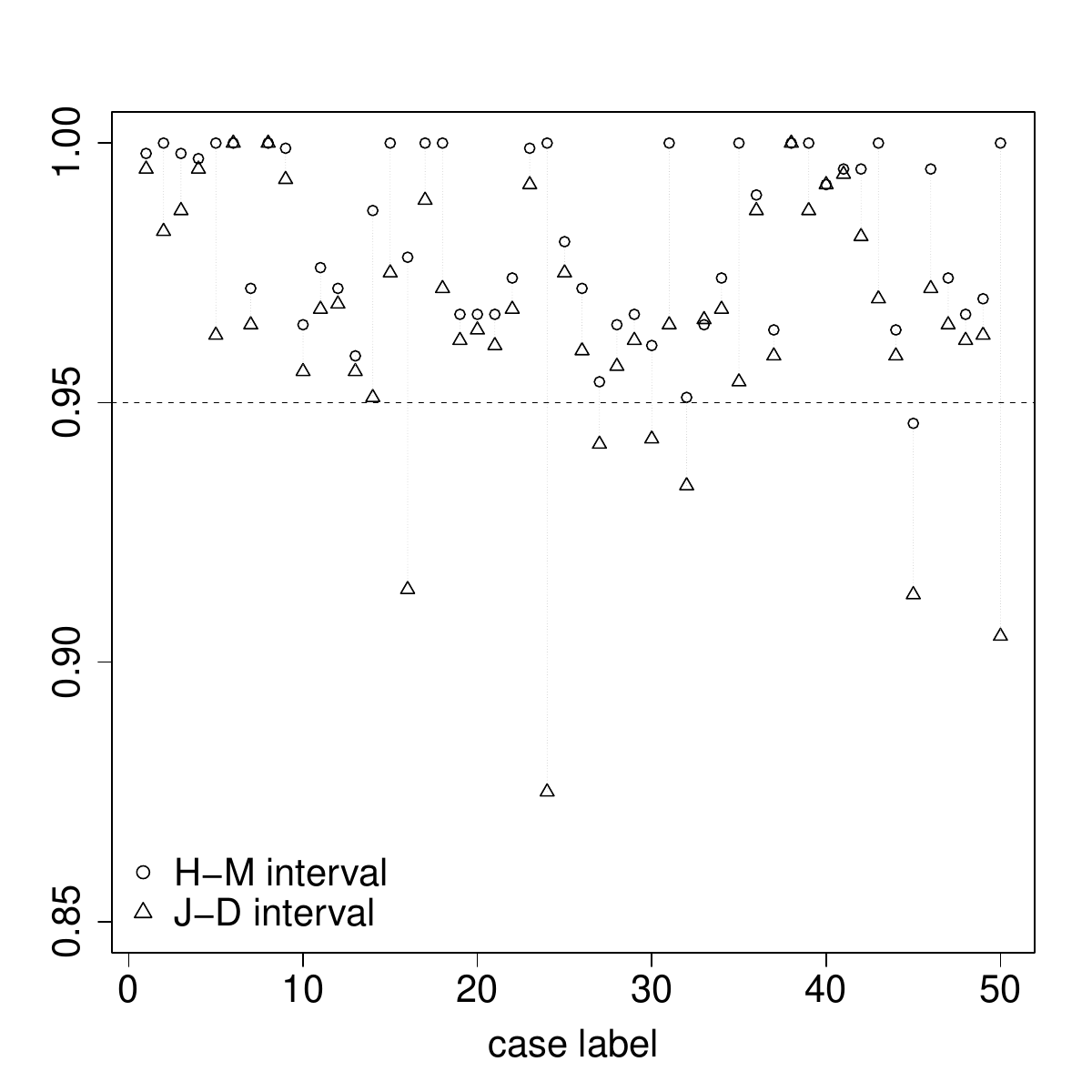}
  \caption{\small Coverage rates for the parameter $\tau.$}
  \label{fig:sub2}
\end{subfigure}
\caption{\small Additional simulation results. In each subfigure, the horizontal axis denotes the simulation case labels, and the vertical axis denotes the coverage rates for the 95\% \cite{Horowitz:2000} interval (denoted as ``H-M interval,'' round dot) and the 95\% \cite{Jiang:2018} interval (denoted as ``J-D interval,'' triangular dot).}
\label{fig:add-simu}
\end{figure}

\end{document}